\documentclass[12pt,reqno]{amsart}

\usepackage{graphicx}
\usepackage{amsthm,amsmath,amsfonts,amssymb,amsxtra,appendix,bookmark,dsfont,bm,color, braket}
\usepackage{cite}
\usepackage{mathtools}

\newtheorem{theorem}{Theorem}[section]
\newtheorem{lemma}[theorem]{Lemma}

\theoremstyle{definition}

\newtheorem{remark}[theorem]{Remark}

\newtheorem{proposition}[theorem]{Proposition}
\newtheorem{corollary}[theorem]{Corollary}

\renewcommand{\epsilon}{\varepsilon}
\newcommand{\E}{\mathcal{E}}

\renewcommand{\phi}{\varphi}
\newcommand{\R}{\mathbb{R}}

\newcommand{\be}{\begin{equation}}

\newcommand{\uR}{\underline{R}}
\newcommand{\uRo}{\underline{R_0}}
\newcommand{\uZ}{\underline{Z}}
\newcommand{\1}{{\ensuremath {\mathds 1} }}

\newcommand\tr{\mathop{\mathrm{tr}}\nolimits} 

\numberwithin{equation}{section}  

\begin{document}

\title[Potential energy surfaces for Kohn-Sham models]{Born-Oppenheimer potential energy surfaces for Kohn-Sham models in the local density approximation}

\author{Yukimi Goto}
\address{RIKEN iTHEMS, Wako, Saitama 351-0198, Japan}
\email{\tt yukimi.goto@riken.jp}

\begin{abstract}
We show that the Born-Oppenheimer potential energy surface in Kohn-Sham theory behaves like the corresponding one in Thomas-Fermi theory up to $o(R^{-7})$ for small nuclear separation $R$.
We also prove that if a minimizing configuration exists, then the minimal distance of nuclei is larger than some constant which is independent of the nuclear charges.
\end{abstract}

\maketitle

\section{Introduction}
We consider a molecule with $N>0$ electrons and $K$ static nuclei at $R_1, \dots, R_K$ of charges $z_1, \dots, z_K >0$.
Density Functional Theory (DFT)~\cite{Levy, LiebDF} tells us that the ground state energy is given by the minimization problem
\begin{align*}
E^\mathrm{GS}_{V_{\uR}} (N) &\coloneqq \inf \left\{ F_\mathrm{LL} (\rho) - \int_{\R^{3}} V_{\uR}(x) \rho (x) \, dx \colon \sqrt \rho \in H^1(\R^3), \int_{\R^3} \rho = N\right\}, \\
& V_{\uR}(x) \coloneqq \sum_{j=1}^K \frac{z_j}{|x-R_j|}, \quad \uR = (R_1, \dots, R_K) \in \R^{3K}.
\end{align*}
Here  $F_\mathrm{LL} (\rho)$ is the Levy-Lieb functional defined by
\begin{align*}
F_\mathrm{LL} (\rho) &\coloneqq  \inf_{\substack{\psi \in \bigwedge^N L^2(\R^3) \\ \|\psi\|_{L^2}=1 \\ \rho_\psi = \rho}}   \left\{\frac{1}{2} \sum_{j=1}^N \int_{\R^{3N}} |\nabla_j \psi(\underline X)|^2 \, d\underline X + \sum_{1 \le i < j \le N}\int_{\R^{3N}} \frac{|\psi(\underline X)|^2}{|x_i-x_j|}\, d\underline X \right\}, \\
 \rho_\psi(x) &\coloneqq N\int_{\R^{3(N-1)}} |\psi(x, x_2, \dots, x_N)|^2 \, dx_2 \cdots dx_N , \quad \underline X = (x_1, \dots, x_N) \in \R^{3N},
\end{align*}
where $\bigwedge^N L^2(\R^3)$ denotes the $N$-particle space of antisymmetric wave functions.
Although DFT gives the exact lowest energy, we usually need suitable approximations.
The Local Density Approximation (LDA) refers to an approximation such as
\[
F_\mathrm{LL} (\rho) \approx \underbrace{\frac{1}{2}\iint_{\R^3 \times \R^3} \frac{\rho(x) \rho(y)}{|x-y|} \, dxdy}_{\eqqcolon D(\rho)} + \underbrace{\int_{\R^3} f(\rho(x)) \, dx}_{\text{local term}}.
\]
For instance, one can obtain the Thomas-Fermi (TF) functional with $f(t) = 3/10(3\pi^2)^{2/3}t^{5/3}$.
More precisely, for $V\colon \R^3 \to \R$ with $V \in L^{5/2} + L^\infty$ and $\rho \ge 0$ we define
\[
\E^\mathrm{TF}_V(\rho) \coloneqq \frac{3}{10}(3\pi^2)^{2/3} \int_{\R^3} \rho(x)^{5/3} \, dx - \int_{\R^3} \rho(x) V(x) \, dx + D(\rho),
\]
and its energy is
\[
E_V^\mathrm{TF} (n) \coloneqq \inf \left\{ \E^\mathrm{TF}_V(\rho) \colon \rho \ge 0,  \int_{\R^3} \rho \le n\right\}.
\]
It is well-known that the unique minimizer $\rho_V^\mathrm{TF}$ exists for any $n>0$ (see, e.g.,~\cite{LiebSimon1977, LiebTF}).
We note that the Levy-Lieb functional includes the kinetic energy and electron-electron repulsive interaction, and TF theory neglects the exchange-correlation energy.
On the other hand, the kinetic energy can be written by $\tr [(-\Delta/2)\gamma]$ with a density-matrix $\gamma \in \mathcal{DM}$ having $\tr \gamma = N$, where
\[
\mathcal{DM} \coloneqq \left\{ \gamma \colon 0 \le \gamma \le 1, \gamma = \gamma^\dagger,  \tr (-\Delta \gamma) < \infty\right\}. 
\]
For a trace operator $\gamma$, its density is $\rho_\gamma(x) \coloneqq \gamma (x, x)$ with  Hilbert-Schmidt kernel $\gamma (x, y) = \sum_{j \ge 1} \lambda_j \phi_j(x) \phi_j(y)^\ast$, where $\gamma \phi_j = \lambda_j \phi_j$.
The (extended) Kohn-Sham model is given by
\begin{align*}
\E_V^\mathrm{KS}(\gamma) &\coloneqq \tr \left[\left(-\frac{1}{2} \Delta - V \right)\gamma\right]  + D(\rho_\gamma) -E_\mathrm{xc}(\rho_\gamma),  \\
E_V^\mathrm{KS}(n) &\coloneqq \inf \left\{ \E_V(\gamma) \colon \gamma \in \mathcal{DM}, \tr \gamma = n \right\},
\end{align*}
where $E_\mathrm{xc}$ is the exchange-correlation energy of the form
\begin{align*}
-E_\mathrm{xc}(\rho) \coloneqq \min_{\substack{\rho = \sum_j \lambda_j \rho_j \\ \sum_j \lambda_j = 1 \\ \sqrt{\rho_j} \in H^1(\R^3) \\ \int_{\R^3} \rho_j = n}} \sum_j \lambda_j F_\mathrm{LL}(\rho_j) - \inf_{\substack{\gamma \in \mathcal{DM} \\ \rho_\gamma = \rho \\ \tr \gamma =n}} \tr\left(-\frac{\Delta}{2} \gamma\right) -D(\rho).
\end{align*}
Then the Kohn-Sham energy is exact, i.e., $E_{V_{\uR}}^\mathrm{GS}(n) = E_{V_{\uR}}^\mathrm{KS}(n)$.
We use an approximate $E_\mathrm{xc}$ called the LDA exchange-correlation functional as
\begin{equation}
\label{def.LDA}
E_\mathrm{xc}(\rho) \approx E_\mathrm{xc}^\mathrm{LDA}(\rho) \coloneqq  \int_{\R^3} g(\rho(x)) \, dx,
\end{equation}
and introduce the Kohn-Sham LDA model
\begin{align*}
\E_V(\gamma) &\coloneqq \tr \left[\left(-\frac{1}{2} \Delta - V \right)\gamma\right]  + D(\rho_\gamma) -E_\mathrm{xc}^\mathrm{LDA}(\rho_\gamma),  \\
E_V(n) &\coloneqq \inf \left\{ \E_V(\gamma) \colon \gamma \in \mathcal{DM}, \tr \gamma = n \right\}.
\end{align*}
The following assumptions will be needed throughout the paper.
In (\ref{def.LDA}), the function $g \colon \R_+ \to \R_+$ is twice differentiable and satisfies
\begin{equation}
\begin{split}
\label{def.condition}
g(0) &= 0, \\
g' &\ge 0, \\
\exists 0 < \beta_- \le \beta_+ \le \frac{2}{5}& \quad \sup_{t \in \R_+}\frac{|g'(t)|}{t^{\beta_-} + t^{\beta_+}} < \infty, \\
\exists 1 \le \alpha < \frac{3}{2} &\quad \limsup_{t \to 0+} \frac{g(t)}{t^\alpha} >0.
\end{split}
\end{equation}
For instance, the LDA exchange functional $g^\mathrm{LDA}(\rho) = (3/4) (3/\pi)^{1/3} \rho^{4/3}$ satisfies (\ref{def.condition}).

Mathematically, the choice $g_\mathrm{LO} (\rho) = 1.45 \rho^{4/3}$ gives a lower bound of $E_\mathrm{xc}(\rho)$~\cite{LewinLieb}, and it has been shown in~\cite{LewinLiebSeiringer} that a quantitative estimate exists between the grand canonical Levy-Lieb energy and the energy of the uniform electron gas, $\int g_\mathrm{UEG} (\rho(x))) \, dx$, containing the kinetic and exchange-correlation energy.
The function $g_\mathrm{UEG}$ behaves like $g_\mathrm{UEG} \sim c_1\rho^{5/3} -c_2\rho^{4/3}$, where the first term can be interpreted as the kinetic energy.
Thus the conditions (\ref{def.condition}) are not so restrictive.

Under the conditions, it has been shown in~\cite{KS} that the Kohn-Sham energy $E_{V_{\uR}} (N)$ has a minimizer (ground state) $\gamma_0$ if $N \le Z \coloneqq \sum_{j=1}^K z_j$.

In this paper, we will investigate the behavior of the potential energy surface at short internuclear distance $R_\mathrm{min} \coloneqq \min_{i \neq j} |R_i - R_j| \to 0$.
Let $U_{\uR} \coloneqq \sum_{i < j} z_iz_j|R_i-R_j|^{-1}$ be the nucleus-nucleus interaction.
Then the Born-Oppenheimer potential energy surfaces are defined as
\begin{align*}
D^\mathrm{TF} (\uZ, \uR) &\coloneqq E^\mathrm{TF}_{V_{\uR}}(Z) - \sum_{j=1}^K E^\mathrm{TF}_{z_j/|x-R_j|}(z_j) + U_{\uR}, \\
D (\uZ, \uR) &\coloneqq E_{V_{\uR}}(Z) - \sum_{j=1}^K E_{z_j/|x-R_j|}(z_j) + U_{\uR}.
\end{align*}

In fact, the atomic energies $E^\mathrm{TF}_{z_j/|x-R_j|}(z_j)$ and $E_{z_j/|x-R_j|}(z_j)$ are independent of the nuclear position $R_j$ since translation invariance of the functionals, and thus their ground state densities are obtained by the translation of the densities, $\rho_{z_j}^\mathrm{TF}$ and $\rho_{z_j}$, for $E^\mathrm{TF}_{z_j/|x|}(z_j)$ and $E_{z_j/|x|}(z_j)$.
In~\cite{BrezisLiebTF}, Brezis and Lieb showed that  $\lim_{l \to \infty } D^\mathrm{TF}(l^3\uZ,  \uR) = \lim_{l \to \infty } l^7D^\mathrm{TF}(\uZ,  l\uR)  \eqqcolon \Gamma(\uR) >0$ for a certain $\Gamma(\uR)$ which is independent of all $z_j$.
Although~\cite{BrezisLiebTF} proved that  $\Gamma(\uR) = D^\mathrm{TF}_\infty R^{-7}$ for two atoms separated by $R=|R_1 -R_2|$, the exact value $D^\mathrm{TF}_\infty$ is not known.
Recently, Solovej has conjectured in~\cite{SolovejTF} that for homonuclear ($z_1 = z_2 = z/2$) diatomic molecules
\[
\limsup_{z \to \infty} \left|D^\mathrm{GS}(\uZ, \uR) - R^{-7}D^\mathrm{TF}_\infty \right| = o(R^{-7}), \quad \text{as } R \to 0,
\]
where $D^\mathrm{GS}(\uZ, \uR)$ stands for the Born-Oppenheimer potential energy surface of the ground state energy $E^\mathrm{GS}_{V_{\uR}} (Z)$.
Results on the opposite regime $R \to \infty$ are also known: for neutral atoms in the quantum theory, van der Waals interaction law $D^\mathrm{GS}(\uZ, \uR) \approx -R^{-6}$ exists for large separation $R$~\cite{LT1986, Anapolitanos, AnapolitanosSigal}.
Furthermore, if the influence of retardation  effects is taken into account, then the long-range interaction becomes $-R^{-7}$~\cite{CasimirPolder}.

In reduced Hartree-Fock (rHF) theory, which is obtained by neglecting the exchange-correlation $E_\mathrm{xc}$ in the Kohn-Sham functional, Solovej's conjecture is settled by Samojlow in his Ph.D thesis~\cite{Samojlow}.

Our main result is a generalization of Samojlow's result to the case of $K \ge 2$ nuclei with the LDA exchange-correlation.

\begin{theorem}
\label{thm.main}
Let $z_\mathrm{max} = \max_{1\le i \le K} z_i$ and $z_\mathrm{min} = \min_{1\le i \le K} z_i$.
If $z_\mathrm{min} \ge 1$ and $z_\mathrm{min} \ge \delta_0 z_\mathrm{max}$ for some $\delta_0 >0$, then there exists $\epsilon >0$ such that for any $R_\mathrm{min} \in (0, 4]$ 
\[
\left|  D(\uZ, \uR) - D^\mathrm{TF}(\uZ, \uR) \right| \le C R_\mathrm{min}^{-7 + \epsilon}.
\]
Moreover, it follows that as $R_\mathrm{min} \to 0$
\begin{equation}
\label{thm.BO-curve}
\limsup_{\substack{z_\mathrm{min} \ge \delta_0  z_\mathrm{max}\\ z_\mathrm{min} \to \infty}} \left| D(\uZ, \uR) - \Gamma(\uR) \right| = o(R_\mathrm{min}^{-7}).
\end{equation}
\end{theorem}

\begin{corollary}
\label{cor.main}
We assume that there is a constant $\delta_0$ such that $z_\mathrm{min} \ge \delta_0 z_\mathrm{max}$.
If there exists $\uRo = (R_0^{(1)}, \dots, R_0^{(K)})$ such that $\inf_{\uR} (E_{V_{\uR}}(Z) + U_{\uR}) = E_{V_{\uRo}}(Z) + U_{\uRo}$, then
\begin{equation}
R_\mathrm{M} \coloneqq \min_{i \neq j}|R_0^{(i)} - R_0^{(j)}| \ge C_0,
\end{equation}
for some constant $C_0 > 0$ independently of the nuclear charges. 
\end{corollary}

In~\cite{CattoLions1, CattoLions2, CattoLions3, CattoLions4}, the existence of optimal $\uR$ for TF-type models was settled, but we have not been able to prove that  for our model.
Hence, although we believe it is true, the existence of such a configuration remains open.

\begin{remark}
The assumptions (\ref{def.condition}) might be slightly loosened, since for a more general function $g \colon \R_+ \to \R$ (see~\cite[Eq.~(25)--(28)]{KS}) the Kohn-Sham energy $E_{V_{\uR}}(Z)$ has a minimizer.
In particular, the optimal bound of $\beta_+$ in (\ref{def.condition}) is presumably much closer to $2/3$.
\end{remark} 

\begin{remark}
It is conjectured that $C_1 \le R_\mathrm{min} \le R_\mathrm{max} \coloneqq \max_{i \neq j}|R_i -R_j| \le C_2$ for some universal constants $C_1, C_2>0$ if a minimizing configuration exists in the quantum theory.
Hence Theorem~\ref{thm.main} suggests that the energy of interaction behaves like $R^{-7}$ if $R \lesssim r =$ the interatomic distance and $-R^{-6}$ at infinity.
\end{remark}

\begin{remark}
An important extension is the Hartree-Fock theory which approximates the exchange-correlation by
\[
E_\mathrm{xc}(\rho_\gamma) \approx X(\gamma) \coloneqq \iint_{\R^3 \times \R^3} \frac{|\gamma(x, y)|^2}{|x-y|} \, dxdy.
\]
Unfortunately, our method does not work for Hartree-Fock theory because $X(\gamma)$ is non-local.
The main difficulty is that the localization error of the energy,
\[
 \int_{B(R_i, r)} \int_{B(R_j, r) } \frac{|\gamma(x, y)|^2}{|x-y|} \, dxdy,
\]
 has to be dominated by $o(r^{-7})$ independently of the nuclear charges for small $r$.
This problem does not arise for the LDA term $E_\mathrm{xc}^\mathrm{LDA}(\rho_\gamma)$ if $r$ is small enough.
\end{remark}

The proof of Theorem~\ref{thm.main} follows the strategy inspired by~\cite{Samojlow}.
Indeed, the main idea is to compare with TF theory, and one of the key ingredients is the Sommerfeld estimate for molecules.
In the case $K=2$, the Sommerfeld estimate was shown in~\cite{Samojlow}, and the proof has been extended to the $K > 2$ case in~\cite{GotorHF}.
In the present article, we generalize certain bounds~\cite{Samojlow} and~\cite{GotorHF} used on the difference between the considered theory and TF theory so that the Sommerfeld estimates can also be used in our case.
From a technical point of view, the non-linearity and non-convexity of the exchange-correlation term are the main mathematical difficulties in studying the Kohn-Sham LDA model.
These are also the reason why conditions (\ref{def.condition}) are different from the one in~\cite{KS}.

This article is organized as follows.
In Section~\ref{gs},  we derive some standard properties for ground states.
Besides, we study a semi-classical analysis in Kohn-Sham theory.
In Section~\ref{screened}, we provide the comparison estimates of the screened potentials, which allows us to control the difference between a ground state density in Kohn-Sham theory with a minimizer of an outer TF functional.
Due to the exchange-correlation, these analyses are always more involved than rHF theory, even in the atomic $K=1$ case. 
Hence the results in Sect.~\ref{gs} and Sect.~\ref{screened} are also some of our contributions and novelties in this paper.
The proof of Theorem~\ref{thm.main} is given in Section~\ref{proof.main} using Solovej's iterative argument introduced in~\cite{Solovej}.
In particular, we study the energy contributions of the densities away from nuclei for both Kohn-Sham and TF theories.
Finally, we prove Corollary~\ref{cor.main}, which is a straightforward consequence of Theorem~\ref{thm.main}.

\section*{Conventions}
We will denote by $\rho^\mathrm{TF}$, $\rho^\mathrm{TF}_{z_j}$, and $\rho_{z_j}$ the minimizers of $E^\mathrm{TF}_{V_{\uR}}(Z)$, $E^\mathrm{TF}_{z_j/|x-R_j|}(z_j)$, and $E_{z_j/|x-R_j|}(z_j)$, respectively.
For $N \le Z$, we denote minimizers for $E_{V_{\uR}}(N)$ and $E_{V_{\uR}}(Z)$ by the same $\gamma_0$ when no confusion can arise, and write its density $\rho_0$ for short. Then we introduce here the screened potentials defined by
\begin{align*}
\Phi_{r}(x) &\coloneqq V_{\uR}(x)-  \int_{A_r^c} \frac{\rho_0(y)}{|x-y|} \, dy,\\
\Phi^\mathrm{TF}_{r}(x) &\coloneqq V_{\uR}(x)-  \int_{A_r^c} \frac{\rho^\mathrm{TF}(y)}{|x-y|} \, dy, \\
\Phi_{j, r}(x) &\coloneqq  z_j|x-R_j|^{-1}-  \int_{|x-R_j| < r} \frac{\rho_{z_j}(y)}{|x-y|} \, dy, \\
\Phi_{j, r}^\mathrm{TF}(x) &\coloneqq  z_j|x-R_j|^{-1}-  \int_{|x-R_j| < r} \frac{\rho^\mathrm{TF}_{z_j}(y)}{|x-y|} \, dy.
\end{align*}
where $A_r^c$ stands for the complement of $A_r = \{x \in \R^3 \colon |x-R_j| > r \, \text{ for all } j=1, \dots, K\}$.
Besides, we will use the standard notation
\[
D(f, g) \coloneqq \frac{1}{2}\iint_{\R^3 \times \R^3} \frac{f(x)g(y)}{|x-y|} \, dxdy.
\]
Our proofs of the results in this paper also work for atomic Kohn-Sham theory with slight modifications.
For instance, the quantity $R_\mathrm{min}/4$ is replaced by $1$ in that case.

\section{Properties of the ground state}
\label{gs}
In this section, we assume $N \le Z$, and $\gamma_0$ denotes a minimizer for $E_{V_{\uR}}(N)$.
First, we show some a-priori bounds for the ground state $\gamma_0$.
\begin{proposition}
\label{prop.apriori}
For any density matrix $\gamma \in \mathcal{DM}$ it follows that for any $\epsilon >0$
\begin{equation}
\label{eq.exchange}
E_\mathrm{xc}^\mathrm{LDA}(\rho_\gamma)  \le \epsilon  \int \rho_\gamma^{5/3} +  2c_\epsilon \tr \gamma,
\end{equation}
where $c_\epsilon = \max\{ 1, \epsilon^{-3/2}\}$.

Moreover, we have

\begin{equation}
\label{eq.kinetic}
0 \ge E_{V_{\uR}}(N) \ge \frac{1}{4} \tr (-\Delta \gamma_0)  - C\sum_{j=1}^Kz_j^{7/3}.
\end{equation}
\end{proposition}

\begin{proof}
By our assumption, we have
\[
E_\mathrm{xc}^\mathrm{LDA}(\rho_\gamma) \le \int (\rho_\gamma^{1+\beta_+} + \rho_\gamma^{1+\beta_-}).
\]
Using H\"older's inequality, we see that
\[
\int \rho_\gamma^{1+ \beta_{\pm}} \le \left(\int \rho_\gamma^{5/3}\right)^{3\beta_{\pm}/2}  \left(\int \rho_\gamma \right)^{1-3\beta_{\pm}/2}.
\]
Now we use the inequality $a^{\alpha} b^{\beta} \le \epsilon \alpha a + \beta \epsilon^{-\alpha \beta^{-1}}b$ for arbitrary $a, b, \epsilon  >0$, and $0 < \alpha < 1$, $0 < \beta < 1$ such that $\alpha + \beta = 1$.
This follows from inserting $x= a/b$ in the simple inequality $x^{\alpha} \le \epsilon \alpha x +(1-\alpha) \epsilon^{-\alpha (1-\alpha)^{-1}}$.
Then (\ref{eq.exchange}) follows.
On the other hand, by the Lieb-Thirring inequality, we have
\[
\tr (-\Delta \gamma_0) \ge C \int_{\R^3} \rho_0(x)^{5/3} \, dx.
\]
In addition, we know $0 \ge E_{V_{\uR}}(N)$~\cite[Lem.~1]{KS}.
Together with these results,  we obtain
\begin{align*}
0 \ge E_{V_{\uR}}(N) &\ge  \frac{1}{4}\tr (-\Delta \gamma) + C^{-1}\int \rho_0^{5/3} - \int V_{\uR} \rho_0 + D(\rho_0) - C Z \\
&\ge \frac{1}{4}\tr (-\Delta \gamma) -C\sum_{j=1}^Kz_j^{7/3},
\end{align*}
where we have used the bound on the Thomas-Fermi energy $E_{V_{\uR}}^\mathrm{TF}(Z) \ge -c\sum_{j=1}^Kz_j^{7/3}$.
This shows (\ref{eq.kinetic}).
\end{proof}

The following lemma is the first step towards a proof of the universal bound of the Born-Oppenheimer energy.
\begin{lemma}[Initial step]
\label{thm.initial}
It follows that
\begin{equation}
\label{eq.semicl}
 E_{V_{\uR}} (N) \ge \E_{V_{\uR}}^\mathrm{TF}(\rho^\mathrm{TF}) +D(\rho_0 - \rho^\mathrm{TF})  - CZ^{25/11}.
\end{equation}
Moreover, there is a universal constant $C_1 > 0$ such that for any $r \in (0, R_\mathrm{min}/4]$
\begin{equation}
\label{ineq.initial}
\sup_{x \in \partial A_r} |\Phi_r^\mathrm{TF}(x) - \Phi_r(x)| \le C_1 Z^{\frac{49}{36}-a}r^{1/12},
\end{equation}
where $a = 1/198$.
\end{lemma}

This lemma allows us to control $x$ near the nuclei since $Z^{49/36}r^{1/12} \le r^{-4}$ for $r \le Z^{-1/3}$.

\begin{proof}[Proof of Lemma~\ref{thm.initial}]
We bound $\E(\gamma_0)$ from above and below.
It is easy to see that $E_{V_{\uR}}(N) \le E_{V_{\uR}}^\mathrm{rHF} (N) \coloneqq \inf \{\E_{V_{\uR}}(\gamma) + E_\mathrm{xc}^\mathrm{LDA}(\rho_\gamma) \colon \gamma \in \mathcal{DM}, \tr \gamma =N\}$, and the upper bound $E_{V_{\uR}}^\mathrm{rHF} (N) \le \E_{V_{\uR}}^\mathrm{TF}(\rho^\mathrm{TF}) + CZ^{25/11}$ has been shown in~\cite[Eq.~(5.2)]{GotorHF}.

Inserting $\epsilon = Z^{-8/15}$ in Proposition~\ref{prop.apriori}, we obtain
\[
E_\mathrm{xc}^\mathrm{LDA}(\rho_0^{1+ \beta_{\pm}}) \le CZ^{\frac{9}{5}}.
\]
Let $\phi^\mathrm{TF} \coloneqq V_{\uR} - \rho^\mathrm{TF}\star|\cdot|^{-1}$ be the TF potential for $E_{V_{\uR}}^\mathrm{TF}(Z)$.
Then we have
\begin{align*}
\E_{V_{\uR}} (\gamma_0) &\ge \tr \left[ \left(-\frac{\Delta}{2}  -V_{\uR} \right)\gamma_0\right] + D(\rho_0) -CZ^{\frac{9}{5}}\\
&= \tr \left[ \left(-\frac{\Delta}{2} \left(1 -\epsilon\right) -\phi^\mathrm{TF}\star g^2 \right)\gamma_0\right] + D(\rho_0 - \rho^\mathrm{TF}) -D(\rho^\mathrm{TF}) \\
&\quad + \tr \left[ \left(-\frac{\Delta}{2} \epsilon- (\phi^\mathrm{TF} - \phi^\mathrm{TF} \star g^2) \right)\gamma_0\right] -CZ^{\frac{9}{5}} ,
\end{align*}
for arbitrary $\epsilon>0$ and $g$.
Now we use coherent states as in~\cite{SolovejIC}.
For any $s>0$ we take the function $g \colon \R^3 \to \R$ such that $g(x) = 0$ if $|x|>s$ and $g(x) = (2\pi s)^{-1/2}|x|^{-1} \sin(\pi |x|/s)$ if $|x| \le s$.
Then it holds that
\[
0 \le g \le 1, \quad \int g^2 = 1, \quad \int |\nabla g|^2 = \left( \frac{\pi}{s} \right)^2.
\]
The coherent states associated $g$ is given by $f_{k, y} (x) = \exp(ik\cdot x) g(x-y)$ for $k, y \in \R^3$.
Let $\pi_{k, y}$ be the projection in $L^2(\R^3)$ onto $f_{k, y}$, i.e., $(\pi_{k, y} \psi) (x) = f_{k, y}\langle f_{k, y}, \psi \rangle$ for $\psi \in L^2(\R^3)$.
Then from the resolution of the identity and representation of the kinetic energy~\cite[Thm.~12.8 \& 12.9]{LiLo}, we have
\begin{align*}
\tr &\left[\left(-\frac{\Delta}{2} \left(1 -\epsilon\right) -\phi^\mathrm{TF}\star g^2\right)\gamma_0 \right] \\
&=  (2\pi)^{-3} \iint\, dk dy \left(\frac{k^2}{2} \left(1  -\epsilon\right) -\phi^\mathrm{TF} (y)\right)\tr(\pi_{k, y}\gamma_0) -\pi^2(2s^2)^{-1}N \\
&\ge (2\pi)^{-3} \iint_{\frac{k^2}{2} \left(1  -\epsilon\right) -\phi^\mathrm{TF} (y) < 0} \, dk dy \left(\frac{k^2}{2} \left(1  -\epsilon\right) -\phi^\mathrm{TF} (y)\right)-\pi^2(2s^2)^{-1}N \\
&= -2^{3/2}(15\pi^2)^{-1} \left(1 -\epsilon\right)^{-3/2} \int_{\R^3}\phi^\mathrm{TF}(x)^{5/2} \, dx - \pi^2(2s^2)^{-1}N.
\end{align*}
On the other hand, the Lieb-Thirring inequality leads to that
\[
\tr \left[ \left(-\frac{\Delta}{2} \epsilon- (\phi^\mathrm{TF} - \phi^\mathrm{TF} \star g^2) \right)\gamma_0\right]
\ge -C\epsilon^{-3/2}\|[\phi^\mathrm{TF} - \phi^\mathrm{TF} \star g^2]_+ \|_{L^{5/2}}^{5/2}.
\]
Optimizing over $\epsilon$, we see that
\begin{align*}
-2^{3/2}(15\pi^2)^{-1}&(1-\epsilon)^{-3/2}  \int_{\R^3}\phi^\mathrm{TF}(x)^{5/2} \, dx   -C\epsilon^{-3/2}\|[\phi^\mathrm{TF} - \phi^\mathrm{TF} \star g^2]_+ \|_{L^{5/2}}^{5/2}\\
& \ge -2^{3/2}(15\pi^2)^{-1}  \int_{\R^3}\phi^\mathrm{TF}(x)^{5/2} \, dx  -C  \|\phi^\mathrm{TF}\|_{L^{5/2}} \|[\phi^\mathrm{TF} - \phi^\mathrm{TF} \star g^2]_+ \|_{L^{5/2}}^{3/2}.
\end{align*}
By the TF equation $2^{-1}(3\pi^2)^{2/3}(\rho^\mathrm{TF})^{2/3} = \phi^\mathrm{TF}$, it follows that
\[
\int \phi^\mathrm{TF}(x)^{5/2} \, dx \le CZ^{7/3}
\]
and
\[
-2^{3/2}(15\pi^2)^{-1} \int[\phi^\mathrm{TF} ]_+^{5/2} - D(\rho^\mathrm{TF}) = \E^\mathrm{TF}(\rho^\mathrm{TF}).
\]
Next, we note that $V_{\uR} - V_{\uR}  \star g^2 \ge 0$ since $V_{\uR} $ is superharmonic.
Then we have
\[
\|[\phi^\mathrm{TF} - \phi^\mathrm{TF} \star g^2]_+\|_{L^{5/2}}^{5/2} \le \int |V_{\uR} - V_{\uR}  \star g^2|^{5/2} \le CZ^{5/2}s^{1/2}.
\]
Here we have used 
\[
V_{\uR} - V_{\uR}  \star g^2 \le \sum_{j=1}^K z_j (|x-R_j|^{-1} \1(|x-R_j| \le s)).
\]
Together with these results, we have

\[
 E_{V_{\uR}} (N) \ge \E^\mathrm{TF}(\rho^\mathrm{TF}) +D(\rho_0 - \rho^\mathrm{TF}) -Cs^{-2}Z - CZ^{12/5}s^{1/5} -CZ^{\frac{9}{5}}.
\]

Optimizing over $s > 0$, we conclude that (\ref{eq.semicl}) and hence
\[
D(\rho_0 - \rho^\mathrm{TF}) \le CZ^{\frac{25}{11}}.
\]
We use the following estimate taken from~\cite[Lem.~12]{mullerIC}.
\begin{lemma}[Coulomb estimate]
For any $f \in L^{5/3} \cap L^{6/5}(\R^3)$ and for any $x \in \R^3$ it follows that
\[
\left|\int_{|y| < |x|} \frac{f(y)}{|x-y|} \, dy \right| \le C\|f\|_{L^{5/3}}^{5/6}(|x| D(f))^{1/12}.
\]
\end{lemma}

By harmonicity of the functional $\Phi_r^\mathrm{TF} - \Phi_r$, we see that for any $r \in (0, R_\mathrm{min}/4]$
\begin{align*}
\sup_{x \in A_r} | \Phi_r^\mathrm{TF} - \Phi_r |
& \le \sum_{j=1}^K \sup_{|x-R_j| = r} \left|\int_{|y| < r} \frac{\rho_0(y+R_j) - \rho^\mathrm{TF}(y + R_j)}{|x-R_j -y|} \right| \\
& \le C\|\rho_0 - \rho^\mathrm{TF} \|_{L^{5/3}}^{5/6}(r D(\rho_0 - \rho^\mathrm{TF}))^{1/12} \\
& \le CZ^{49/36 - a}r^{1/12},
\end{align*}
which is the desired conclusion.
\end{proof}

\section{Screened potential estimates}
\label{screened}
From now on $\gamma_0$ denotes a minimizer for $E_{V_{\uR}}(Z)$.
we choose the smooth function $\eta_r \colon \R^3 \to [0, 1]$ such that $\1_{A_r} \ge \eta_r \ge \1_{A_{(1+\lambda)r}}$ and partition of unity, $\eta_r^2 + \eta_+^2 + \eta_-^2 = 1$, satisfying
\[
\mathrm{supp} \, \eta_- \subset A_r^c, \quad \mathrm{supp} \, \eta_+ \subset A_{(1-\lambda)r} \cap A_{(1+\lambda)r}^c,
\]
where $\eta_- = 1$ in $A^c_{(1-\lambda)r}$ and 
\[
\sum_{\# = +, -, r} |\nabla \eta_\#|^2 \le C(\lambda r)^{-2}.
\]

Now we introduce the notation
\label{ineq.iterative1}
\begin{equation}
\mathcal{A} \coloneqq \left\{(r,  \beta, \epsilon) \colon \sup_{x \in \partial A_r} |\Phi_r^\mathrm{TF}(x) - \Phi_r(x)| \le \beta r^{-4 + \epsilon}\right\}.
\end{equation}

Our goal in this section is to provide the following universal bound for the screened potential which is the main technical tool.

\begin{lemma}[Screened potential estimate]
\label{lem.potential}
If $z_\mathrm{min} \ge \delta_0 z_\mathrm{max}$ for some $\delta_0$, then there are constants $C_0, \epsilon_1, \delta_1 >0$ such that $(r, C_0, \epsilon_1) \in \mathcal{A}$ for any $r \in (0, (R_\mathrm{min}/4)^{1+ \delta_1}]$.
\end{lemma}

We will prove Lemma~\ref{lem.potential} by using Solovej's bootstrap argument.
The strategy is based on the initial step and the following iterative step.
Lemma~\ref{thm.initial} shows $(r, C, \epsilon) \in \mathcal{A}$ for $r \le Z^{-1/3}$, and we can extend the range of such $r$ up to $\mathcal{O}(1)$ by an iterative procedure. 

\begin{lemma}[Iterative step]
\label{thm.iterative}
Let $\eta = (7+\sqrt{73})/2 \sim 7.772$ and $\xi = (\sqrt 73 - 7)/2 \sim 0.77$.
We put  $r \in [z_\mathrm{min}^{-1/3}, D]$ with some $D \in [ z_\mathrm{min}^{-1/3}, R_\mathrm{min}/4]$, and $\tilde r \coloneqq r^{\xi/(\xi + \eta)}(R_\mathrm{min}/4)^{\eta/(\xi + \eta)}$.
There are universal constants $C_2, \beta_1, \delta_2, \epsilon_2 >0$ such that, if $(s, \beta_1, 0) \in \mathcal{A}$ holds for any $s \in (0, r]$,
then,  $(s, C_2, \epsilon_2) \in \mathcal{A}$ holds for any $s \in [r^{1/(1+\delta_2)}, \min \{r^{(1-\delta_2)/(1+\delta_2)}, \tilde r\}]$.
\end{lemma}

\begin{remark}
The Sommerfeld asymptotic refers to $\phi^\mathrm{TF} (x) \sim 3^4 2^{-3} \pi^2 |x|^{-4}$ for large $|x|$, and  the important thing to our purpose is the next order.
The above $\eta$ and $\xi$ are the solutions of $p^2 -7p = 6$, which comes from comparing $\Delta |x|^{-4}(1+ |x|^p) = 12|x|^{-6} (1+(p^2 -7p +12)|x|^p/12)$ with $(|x|^{-4}(1+|x|^p))^{3/2} \sim |x|^{-6}(1+3|x|^p/2)$.
Our $\xi$ and $\eta$ are needed for large $|x|$ and for $x$ close to $\partial A_r$ respectively.
\end{remark}

To prove Lemma~\ref{thm.iterative}, we collect the properties of elements in $\mathcal{A}$.
\begin{lemma}
\label{step1}
Let $\beta, D \in (0, R_\mathrm{min}/4]$ be some constants.
We assume that $(r, \beta, 0) \in \mathcal{A}$ holds for all $r \le D$.
Then for any $r \in (0, D]$ we have

\begin{equation}
\label{ite1}
\sup_{A_r}\left| \Phi_r\right| \le \frac{C}{r^4},
\end{equation}

\begin{equation}
\label{ite2}
\left|  \int_{A_r^c} (\rho_0 - \rho^\mathrm{TF}) \right| \le \frac{C\beta}{r^3},
\end{equation}

\begin{equation}
\label{ite3}
\int_{A_r} \rho_0 \le \frac{C}{r^3},
\end{equation}

\begin{equation}
\label{ite4}
\int_{A_r} \rho_0^{5/3} \le \frac{C}{r^7},
\end{equation}

\begin{equation}
\label{ite5}
\tr (-\Delta \eta_r \gamma_0 \eta_r) \le C\left(\frac{1}{r^7}+ \frac{1}{\lambda^2 r^5}\right), \quad \text{for any } \lambda \in (0, 1/2].
\end{equation}
\end{lemma}

\begin{proof}[Proof of Lemma~\ref{step1}]
We may split
\[
\Phi_r (x) = \Phi_r (x) - \Phi_r^\mathrm{TF}(x) +\Phi_r^\mathrm{TF}(x).
\]
From the Sommerfeld bound and the relation $\phi^\mathrm{TF} \le \sum_{j=1}^K\phi^\mathrm{TF}_{z_j}$~\cite[Cor.~3.6]{LiebTF}, where $\phi_{z_j}^\mathrm{TF}$ is the TF potential for the density $\rho^\mathrm{TF}_{z_j}$, we can see that for $x \in A_r$
\[
\Phi_r^\mathrm{TF}(x)  = \phi^\mathrm{TF}(x) + \int_{A_r} \frac{\rho^\mathrm{TF}(y)}{|x-y|} \, dy \le Cr^{-4}.
\]
Then (\ref{ite1}) follows from our assumption.

Next, we use the following lemma.
\begin{lemma}
\label{lem.outer}
Let $f_j$ be a continuous harmonic function on $B(R_j, r)^c$ vanishing at infinity and $f \coloneqq \sum_{j=1}^Kf_j$.
Then we have for any $x \in A_r$ with $r \in (0, R_\mathrm{min}/4]$
\[ 
|f(x)| \le \frac{4}{3}r \sup_{\partial A_r} |f| \sum_{j=1}^K |x-R_j|^{-1}.
\]
\end{lemma}

\begin{proof}[Proof of Lemma~\ref{lem.outer}]
We note that  $|x-R_j| |f_{j}(x)| \le r  \sup_{\partial B(R_j, r)}|f_{j}|$ for any $x \in B(R_j, r)^c$ by the maximum principle (see~\cite[Lem.~6.5]{TFDWIC}).
Then we have for any fixed $j$ and $x \in A_r$
\[
\left|\sum_{i \neq j} f_i(x) \right| \le \sup_{\partial A_r} |f| + \frac{r}{R_\mathrm{min} -r} \sup_{\partial B(R_j, r)} |f_{j}|.
\]
Since $f_{j} = f - \sum_{i \neq j} f_{i}$, we see that $\sup_{\partial B(R_j, r)} |f_j| \le (4/3) \sup_{\partial A_r} |f|$ and thus for any $x \in A_r$
\[
\left| f(x) \right| \le \sum_{j=1}^K \frac{r}{|x-R_j|} \sup_{\partial B(R_j, r)} |f_j| \le  \sup_{\partial A_r} |f|\sum_{j=1}^K \frac{4r}{3|x-R_j|}, \quad \forall x \in A_r,
\]
which shows the lemma
\end{proof}

Using Lemma~\ref{lem.outer} with $f = \Phi_r - \Phi_r^\mathrm{TF}$, we have
\[
\left|\int_{A_r^c}(\rho_0 - \rho^\mathrm{TF}) \right| =\lim_{|x| \to \infty} |x|\left|\Phi_r^\mathrm{TF}(x) - \Phi_r(x) \right| \le \frac{4}{3} \beta K r^{-3}.
\]
This shows (\ref{ite2}).
Then (\ref{ite3}) follows from the Sommerfeld bound $\int_{A_r} \rho^\mathrm{TF}\le Cr^{-3}$ and splitting
\[
\int_{A_r} \rho_0 = \int_{A_r} \rho^\mathrm{TF} + \int_{A_r^c} (\rho^\mathrm{TF} - \rho_0) \le Cr^{-3},
\]
where we have used (\ref{ite2}).

Now we introduce the exterior reduced Hartree-Fock model
\[
\E_r^\mathrm{rHF} (\gamma) \coloneqq \E_{\Phi_r}^\mathrm{rHF} (\gamma)  =\tr \left[\left(-\frac{\Delta}{2} - \Phi_r\right)\gamma\right] + D(\rho_\gamma).
\]
Then we can split outsides from insides as follows.
\begin{lemma}
\label{lem.rhf}
For any $r \in (0, R_\mathrm{min}/4], \lambda \in (0, 1/2]$ and for any $0 \le \gamma \le 1$ satisfying
\[
\mathrm{supp} \, \rho_\gamma \subset A_r, \quad \tr \gamma \le \int_{A_r} \rho_0,
\]
it holds that
\[
\E_{V_{\uR}}(\eta_- \gamma_0 \eta_-) + \E_r^\mathrm{rHF}(\eta_r \gamma_0 \eta_r) - \mathcal{R} \le \E_{V_{\uR}}(\gamma_0) \le \E_{V_{\uR}}(\eta_- \gamma_0 \eta_-) + \E_r^\mathrm{rHF}(\gamma),
\]
where
\begin{equation}
\begin{split}
\mathcal{R} &\le C(1+ (\lambda r)^{-2}) \int_{A_{(1-\lambda)r} \cap A^c_{(1+\lambda)r} } \rho_0  +C \lambda r^3 \sup_{A_{(1-\lambda)r}}[\Phi_{(1-\lambda)r}]_+^{5/2}\\
&\quad +  C\left(\tr(-\Delta \eta_r \gamma_0 \eta_r) \right)^{1/2} \left(\int \eta_r \rho_0 \right)^{1/2}
\end{split}
\end{equation}
\end{lemma}

\begin{proof}[Proof of Lemma~\ref{lem.rhf}]
First, we note that $N  \mapsto E_{V_{\uR}}(N)$ is non-increasing  by~\cite[Lem.~1]{KS}.
Since $\eta_-$ and $\rho_\gamma$ have disjoint supports, we obtain
\begin{align*}
\E_{V_{\uR}}(\gamma_0) &\le \E_{V_{\uR}}(\gamma + \eta_- \gamma_0 \eta_-) \\
& = \E_{V_{\uR}}(\eta_- \gamma_0 \eta_-) + \E_{V_{\uR}}(\gamma) + 2D(\eta_-^2 \rho_0, \rho_\gamma) \\
&\le \E_{V_{\uR}}(\eta_- \gamma_0 \eta_-) + \E_r^\mathrm{rHF}(\gamma),
\end{align*}
which is the desired upper bound.

Second, by the IMS formula we see that
\begin{align*}
\E_{V_{\uR}}(\gamma_0) &= \sum_{\#= +, -, r} \left(\E_{V_{\uR}}(\eta_\# \gamma_0 \eta_\#) - \int |\nabla \eta_\#|^2 \rho_0  \right) \\
&\quad +2D(\eta_r^2 \rho_0, (\eta^2_+ + \eta^2_-)\rho_0) + 2D(\eta^2_- \rho_0, \eta_+^2 \rho_0) \\
&\quad -\int \left(g(\rho_0) -\sum_{\#=+, -, r}g(\eta^2_\# \rho_0)\right).
\end{align*}
For the error terms, we have
\[
\sum_{\# = +, -, r}\int |\nabla \eta_\#|^2 \rho_0 \le C(\lambda r)^{-2} \int_{A_{(1-\lambda)r} \cap A_{(1+\lambda)r}^c} \rho_0.
\]
Next, a simple computation shows that
\begin{align*}
\E_{V_{\uR}}&(\eta_r \gamma_0 \eta_r) + 2D(\eta_r^2 \rho_0, (\eta^2_+ + \eta^2_-)\rho_0) \\
& \ge \E_{V_{\uR}}(\eta_r \gamma_0 \eta_r) +2D(\eta_r^2 \rho_0, \1_{A_r^c} \rho_0) \\
&= \E_r^\mathrm{rHF}(\eta_r \gamma_0 \eta_r) -  E_\mathrm{xc}^\mathrm{LDA}(\eta_r^2 \rho_0),
\end{align*}
and
\begin{align*}
\E_{V_{\uR}}&(\eta_+ \gamma_0 \eta_+) + 2D(\eta_+^2 \rho_0,  \eta^2_- \rho_0) \\
& \ge \E_{V_{\uR}}(\eta_+ \gamma_0 \eta_+) + 2D(\eta_+^2 \rho_0,  \1_{A^c_{(1-\lambda)r}} \rho_0)  \\
&= \E_{(1-\lambda)r}^\mathrm{rHF}(\eta_+ \gamma_0 \eta_+) -  E_\mathrm{xc}^\mathrm{LDA}(\eta_+^2 \rho_0).
\end{align*}
We note that
\[
g(\rho_0) - g(\eta_-^2 \rho_0) \le C(\rho_0^{\beta_-} + \rho_0^{\beta_+}) (\eta_+^2 + \eta_r^2) \rho_0,
\]
and, by H\"older's inequality and the Lieb-Thirring inequality, for any $\beta \le 2/5$ and $ 0 \le \chi \le 1$
\begin{align*}
\int_{A_{(1-\lambda)r}} \rho_0^{1+\beta} \chi^2 &\le \left( \int (\chi^2 \rho_0)^{5/3}\right)^{3 \beta /2} \left(\int_{A_{(1-\lambda)r}} \rho_0 \right)^{1-3 \beta /2}
\\
&\le C \left( \tr (- \Delta \chi \gamma_0 \chi)\right)^{3 \beta /2} \left(\int_{A_{(1-\lambda)r}} \rho_0 \right)^{1-3 \beta /2} \\
&\le \frac{1}{8}  \tr (- \Delta \chi \gamma_0 \chi) + C \int_{A_{(1-\lambda)r}} \rho_0.
\end{align*}
In the last inequality, we have used the simple inequality  $a^{\alpha} b^{\beta} \le \epsilon \alpha a + \beta \epsilon^{-\alpha \beta^{-1}}b$ for arbitrary $a, b, \epsilon  >0$, and $0 < \alpha < 1$, $0 < \beta < 1$ such that $\alpha + \beta = 1$ (recall the proof of Proposition~\ref{prop.apriori}).
The Lieb-Thirring inequality with $V = \Phi_{(1-\lambda)r} \1_{\mathrm{supp} \eta_+}$ implies that
\[
\tr \left[ \left(-\frac{\Delta}{4} - \Phi_{(1-\lambda)r} \right) \eta_+ \gamma_0 \eta_+ \right] \ge -C\int [V]_+^{5/2} \ge - C \lambda r^3 \sup_{A_{(1-\lambda)r}}[\Phi_{(1-\lambda)r}]_+^{5/2}.
\]
Together with these estimates, we have the lemma.
\end{proof}

Applying Lemma~\ref{lem.rhf}, we can obtain the kinetic energy estimate.
\begin{lemma}
\label{lem.kinetic}
For all $r \in (0, R_\mathrm{min}/4]$ and all $\lambda \in (0, 1/2]$ it holds that
\begin{align*}
\tr (- \Delta \eta_r \gamma_0 \eta_r) &\le C(1+ (\lambda r)^{-2}) \int_{A_{(1-\lambda)r} } \rho_0 \\
&\quad +C \lambda r^3 \sup_{A_{(1-\lambda)r}}[\Phi_{(1-\lambda)r}]_+^{5/2} 
+ C\sup_{\partial A_r}|r\Phi_r|^{7/3}.
\end{align*}
\end{lemma}

\begin{proof}[Proof of Lemma~\ref{lem.kinetic}]
We use Lemma \ref{lem.rhf} with $\gamma = 0$ and obtain $\E_r^\mathrm{rHF} (\eta_r \gamma_0 \eta_r)\le \mathcal{R}$.
On the other hand, by the Lieb-Thirring inequality and property of the ground state energy of TF theory, we have

\begin{align*}
\E_r^\mathrm{rHF}(\eta_r \gamma_0 \eta_r) 
&\ge  \tr \left(-\frac{\Delta}{4} \eta_r \gamma_0 \eta_r\right) +C^{-1} \int (\eta_r^2 \rho_0)^{5/3}  \\
&\quad- C\sup_{\partial A_r} |r\Phi_{r}|\sum_{j=1}^K \int \eta_r^2 \frac{\rho_0(x)}{|x-R_j|}\, dx +D(\eta_r^2 \rho_0) \\
&\ge  \tr \left(-\frac{\Delta}{4} \eta_r \gamma_0\eta_r\right)  - C \sup_{ \partial A_r}|r\Phi_{r}|^{7/3},
\end{align*}
where we have used Lemma~\ref{lem.outer}.
This completes the proof.
\end{proof}

Combining this with (\ref{ite3}), we deduce from $(1-\lambda)r > r/3$ that
\[
\tr (- \Delta \eta_{r} \gamma_0 \eta_{r}) \le C\left(\lambda^{-2}r^{-5} + r^{-7} \right),
\]
which shows (\ref{ite5})
Replacing $r$ by $r/3$, we learn
\[
\int_{A_r} \rho_0^{5/3} \le \int (\eta^2_{r/3} \rho_0)^{5/3} \le C\tr (- \Delta \eta_{r/3} \gamma_0 \eta_{r/3})\le C\left(\lambda^{-2}r^{-5} + r^{-7} \right),
\]
where we have used the Lieb-Thirring inequality.
Choosing $\lambda = 1/2$, we have (\ref{ite4}).
\end{proof}

With $V_{r}(x) = \1_{A_r}\Phi_{r}(x)$, we denote the exterior Thomas-Fermi functional $\E_{V_r}^\mathrm{TF}(\rho)$ briefly by $\E_{r}^\mathrm{TF}(\rho) $.
The following lemmata are very similar to that of~\cite[Lem.~6.4,~6.6,~6.8]{GotorHF}, but we provide their proofs for the reader's convenience.
\begin{lemma}
\label{step2}
The exterior TF energy $E_{r}^\mathrm{TF}(\tr (\1_{A_r}\gamma_0 \1_{A_r}))$ has a unique minimizer $\rho^\mathrm{TF}_{r}$, which is supported on $A_r$ and satisfies the TF equation
\[
\frac{1}{2}(3\pi^2)^{2/3}\rho^\mathrm{TF}_{r}(x)^{2/3} = [\phi_{r}^\mathrm{TF}(x)-\mu_r]_+
\]
with $\phi^\mathrm{TF}_{r}(x) = V_{r}(x) - \rho_{r}^\mathrm{TF} \star |x|^{-1}$ and a constant $\mu_r \ge 0$. 
Moreover, 
\begin{itemize}
\item[(i)]
If $\mu_r > 0$, then 
\[
\int \rho^\mathrm{TF}_{r} = \int_{A_r}\rho_0.
\]
\item[(ii)]
If $(r, \beta, 0) \in \mathcal A$ holds true for some $\beta$ and any $r \in (0, D]$ with $D \in (0, R_\mathrm{min}/4]$, then
\begin{equation*}
\label{TFminbound}
\int (\rho_{r}^\mathrm{TF})^{5/3} \le Cr^{-7}, \quad \text{for any } r \in (0, D].
\end{equation*}
\end{itemize}
\end{lemma}

\begin{proof}
By $\phi_r^\mathrm{TF} \le V_r$ and the TF equation, $\mathrm{supp} \, \rho^\mathrm{TF}_{r} \subset A_r$ follows.
From the fact that $\inf_{\rho \ge 0 }\E^\mathrm{TF}_{V_{\uR}}(\rho) \ge -C\sum_{j} z_j^{7/3}$ and Lemma~\ref{lem.outer}, we can see
\begin{align*}
0 \ge \E_{V_r}^\mathrm{TF}(\rho_r^\mathrm{TF}) &\ge \frac{3}{10}(3\pi^2)^{2/3} \int (\rho_r^\mathrm{TF})^{5/3} - Cr^{-3}\sum_{j=1}^K\int \rho_r^\mathrm{TF}(x) |x-R_j|^{-1}\, dx + D(\rho_r^\mathrm{TF})\\
&\ge  \frac{3}{5}(3\pi^2)^{2/3}\int (\rho_r^\mathrm{TF})^{5/3} -Cr^{-7},
\end{align*}
which shows (ii).
The rest of the proof was shown in~\cite{LiebSimon1977}.
\end{proof}

\begin{lemma}
\label{step3}
Let $D \in [  z_\mathrm{min}^{-1/3}, R_\mathrm{min}/4]$.
We can choose a universal constant $\beta > 0$ small enough such that, if $(r, \beta, 0) \in \mathcal{A}$ holds for any $r \in [z_\mathrm{min}^{-1/3}, D]$, then $\mu_{r} = 0$ and for any $s \in [r, \tilde r]$ with $\tilde r = r^{\frac{\xi}{\xi + \eta}}(R_\mathrm{min}/4)^{\frac{\eta}{\xi+\eta}}$ it follows that
\begin{align}
\label{step3a}
\sup_{x\in \partial A_s}|\phi_{r}^\mathrm{TF}(x) - \phi^\mathrm{TF}(x)|&\le C(r/s)^\xi s^{-4}, \\
\label{step3b}
\sup_{x\in \partial A_s}|\rho_{r}^\mathrm{TF}(x) - \rho^\mathrm{TF}(x)| &\le C(r/s)^\xi s^{-6}.
\end{align}
\end{lemma}

\begin{proof}
{\bf (Step 1)}:
First, we show that $\mu_r \le C\beta^{1/2} r^{-4}$ and
\begin{equation}
\label{eq.TFcompare}
D(\rho_r^\mathrm{TF} -\rho^\mathrm{TF} \1_{A_r}) \le C\beta r^{-7+\epsilon}.
\end{equation}
Let $\rho_{r, t}  \coloneqq \rho^\mathrm{TF} \1_{A_r \cap A_t^c}$ and $W(x) = \Phi_{r}^\mathrm{TF}(x) - \Phi_{r}(x)$.
Then for any $t \ge r$
\[
\E_r^\mathrm{TF}(\rho_{r, t}) + \mu_r \int \rho_{r, t} \ge \E_r^\mathrm{TF}(\rho_r^\mathrm{TF}) + \mu_r \int \rho_r^\mathrm{TF},
\]
where we have used the fact that $\mu_r \int \rho_r^\mathrm{TF} = \mu_r\int_{A_r} \rho_0$.
By the same method as in the proof of Lemma~\ref{step2}, we can see that
\begin{align*}
 \E_r^\mathrm{TF}(\rho_{r, t}) -  \E_r^\mathrm{TF}(\rho^\mathrm{TF} \1_{A_r})
&= -\E^\mathrm{TF}_{W}(\rho^\mathrm{TF} \1_{A_t}) + \int_{A_t} \Phi_t^\mathrm{TF} \rho^\mathrm{TF} \\
&\le C\beta r^{-7+\epsilon} + Ct^{-7}.
\end{align*}
Since $t \mapsto (3/10)c_\mathrm{TF} t^{5/3} - \phi^\mathrm{TF} t$ takes its minimum at $t = \rho^\mathrm{TF}$, we learn
\begin{align*}
 \E_r^\mathrm{TF}(\rho^\mathrm{TF} \1_{A_r}) - \E_r^\mathrm{TF}(\rho_r^\mathrm{TF}) 
&= \int_{A_r} W\left(\rho^\mathrm{TF} - \rho_r^\mathrm{TF}\right) - D(\rho_r^\mathrm{TF} -\rho^\mathrm{TF} \1_{A_r}) \\
&\quad+ \int_{A_r}\left(\frac{3}{10}c_{\mathrm{TF}}(\rho^\mathrm{TF})^{5/3} - \frac{3}{10}c_{\mathrm{TF}}(\rho_r^\mathrm{TF})^{5/3} -\phi^\mathrm{TF}\rho^\mathrm{TF} + \phi^\mathrm{TF}\rho_r^\mathrm{TF}\right)\\
&\le C\beta r^{-7+\epsilon} - D(\rho_r^\mathrm{TF} -\rho^\mathrm{TF} \1_{A_r}).
\end{align*}
Combining these estimates, we arrive at
\begin{align*}
0\le \mu_r\left(\int_{A_r} \rho_0 -\int \rho_{r, t} \right) 
&\le \E_r^\mathrm{TF}(\rho_{r, t}) -  \E_r^\mathrm{TF}(\rho^\mathrm{TF} \1_{A_r}) +   \E_r^\mathrm{TF}(\rho^\mathrm{TF} \1_{A_r}) - \E_r^\mathrm{TF}(\rho_r^\mathrm{TF})  \\
&\le C(\beta r^{-7+\epsilon} +\beta r^{-7+\epsilon} + t^{-7})  - D(\rho_r^\mathrm{TF} -\rho^\mathrm{TF} \1_{A_r}).
\end{align*}
Choosing $t = \beta^{-1/7} r^{1-\epsilon}$, we have (\ref{eq.TFcompare}).

Since $\phi^\mathrm{TF} \ge \max_j \phi_{z_j}^\mathrm{TF}$~\cite[Thm.~3.4]{LiebTF} and the Sommerfeld bound~\cite[Thm.~5.4]{SolovejIC}, we see $\int_{A_s} \rho^\mathrm{TF} \ge C^{-1}s^{-3}$ for any $s \ge z_\mathrm{min}^{-1/3}$.
We note that
\[
\int_{A_r}\left(\rho^\mathrm{TF} - \rho_0 \right) = \int_{A_r^c}\left(\rho_0 -\rho^\mathrm{TF} \right) \le \beta r^{-3} \le C\beta \int_{A_r} \rho^\mathrm{TF}.
\]
Hence it holds that for $t = \beta^{-1/6} r$
\[
\int_{A_r} \rho_0 -\int \rho_{r, t} \ge \int_{A_t}\rho^\mathrm{TF} -C\beta\int_{A_r}\rho^\mathrm{TF} \ge C^{-1}\beta^{1/2}r^{-3} -C\beta r^{-3}.
\]
Then the conclusion $\mu_r \le C\beta^{1/2} r^{-4}$ follows for $\beta$ sufficiently small.

{\bf (Step 2)}:
We turn to prove $\mu_r=0$.
By the Sommerfeld bound and our assumption, we see
\begin{align*}
 \inf_{\partial A_r} \phi_r^\mathrm{TF} 
 &= \inf_{\partial A_r} \left(\phi^\mathrm{TF} -[\Phi_r^\mathrm{TF} - \Phi_r] +(\rho^\mathrm{TF} \1_{A_r} -\rho_r^\mathrm{TF})\star|x|^{-1} \right) \\
 &\ge C^{-1}r^{-4} -\beta r^{-4+\epsilon} -\sup_{\partial A_r}|(\rho^\mathrm{TF} \1_{A_r} -\rho_r^\mathrm{TF})\star|x|^{-1}|.
\end{align*}
By Step 1 and the Coulomb estimate $f\star |x|^{-1} \le C\| f\|_{L^{5/3}}^{5/7}D[f]^{1/7}$~\cite[Lem.~6.4]{TFDWIC}, we find
\begin{align*}
\sup_{\partial A_r}|(\rho^\mathrm{TF} \1_{A_r} -\rho_r^\mathrm{TF})\star|x|^{-1}|
&\le C\| \rho^\mathrm{TF} \1_{A_r} -\rho_r^\mathrm{TF}\|_{L^{5/3}}^{5/7}D[\rho^\mathrm{TF} \1_{A_r} -\rho_r^\mathrm{TF}]^{1/7}  \\
&\le C\beta^{1/7}r^{-4+\epsilon}.
\end{align*}
Hence if $\beta >0$ is small enough then we deduce from Step 1 that
\[
 \inf_{\partial A_r} \phi_r^\mathrm{TF}  > C^{-1}r^{-4} \ge \mu_r.
\]
Then by the Sommerfeld estimate for molecules~\cite[Lem.~4.1]{GotorHF} we see
\[
C^{-1}\mu_r^{3/4}(1+ a(r))^{-1/2} \le \lim_{|x| \to \infty}|x| \phi_r^\mathrm{TF} (x) = \int_{A_r} \rho_0 - \int \rho_r^\mathrm{TF},
\]
where $a(r) \coloneqq \sup_{\partial A_r} (\sqrt{c_\mathrm{S} (\phi_r^\mathrm{TF} )^{-1}r^{-4}} -1)$.
This shows $\mu_r =0$ by Lemma~\ref{step2}.

{\bf (Step 3)}:
Let $D_j \coloneqq \min_{i \neq j}|R_i -R_j|/2$.
Using the Sommerfeld bound for molecules~\cite[Lem.~4.1~\&~Lem.~4.2]{GotorHF}, we have for any $x \in A_r \cap \Gamma_j$
\[
|\phi^\mathrm{TF}(x) - \phi_r^\mathrm{TF}(x)| \le C|x-R_j|^{-4}\left(\left(\frac{|x-R_j|}{D_j}\right)^{\eta} + \left(\frac{r}{|x-R_j|}\right)^\xi \right),
\]
where $\xi = (-7 +\sqrt{73})/2$ and $\eta = (7 +\sqrt{73})/2$.
Since $s \le \tilde r$ implies $(s/D_j)^{\eta} \le C(r/s)^{\xi}$, we have (\ref{step3a}).
Then (\ref{step3b}) follows from $(1+a)^{3/2} \le 1+a((1+b)^{3/2}-1)/b$ for any $a \in [0, b]$.
\end{proof}

\begin{lemma}
\label{step4}
Let $\beta > 0$ be as in Lemma \ref{step3} and $D \in [ z_\mathrm{min}^{-1/3}, R_\mathrm{min}/4]$. We assume that $(r, \beta, 0) \in \mathcal{A}$ for any $r \in (0, D]$. Then, if $r \in [z_\mathrm{min}^{-1/3}, D]$, we have
\begin{equation}
\label{eq.step4}
 \E_{r}^\mathrm{TF}(\rho_{r}^\mathrm{TF}) + D(\eta_r^2 \rho_0 - \rho_{r}^\mathrm{TF}) -Cr^{-7 +1/3}
 \le
\E_r^\mathrm{rHF}(\eta_r\gamma_0 \eta_r) 
\le \E_{r}^\mathrm{TF}(\rho_{r}^\mathrm{TF}) + Cr^{-7+ 1/3},
\end{equation}
and
\[
D(\rho_{r}^\mathrm{TF} - \1_{A_r}\rho) \le Cr^{-7+1/3}.
\]
\end{lemma}

\begin{proof}
\underline{\textit{Upper Bound.}}
We will prove that
\begin{equation*}
\label{eq.upper4}
\E_r^\mathrm{rHF}(\eta_r\gamma_0\eta_r) \le \E_{r}^\mathrm{TF}(\rho_{r}^\mathrm{TF}) + Cr^{-7}(r^{2/3} + \lambda^{-2} r^2 + \lambda).
\end{equation*}
Let $s \le r$ be a constant to be chosen later.
We take the function $g$ and projection $\pi_{k, y}$ as in Lemma~\ref{thm.initial}, and define
\[
\widetilde \gamma \coloneqq (2\pi)^{-3} \iint_{\frac{k^2}{2} - V_r'(y) \le 0} \pi_{k, y} \, dy dk,
\] 
with $V_{r}' \coloneqq \1_{A_{r+s}} \phi_{r}^\mathrm{TF}$.
Since $\mu_r = 0$ by Lemma \ref{step3} and the TF equation in Lemma~\ref{step2}, we can see
\[
\rho_{\widetilde \gamma} = (\1_{A_{r+s}} \rho_{r}^\mathrm{TF})\star g^2.
\]
Since $\rho_{\widetilde \gamma}$ is supported in $A_r$ and
\[
\tr \widetilde \gamma = \int \rho_{\widetilde \gamma} = \int_{A_{r+s}}  \rho_{r}^\mathrm{TF} 
\le
\int \rho_{r}^\mathrm{TF}  \le \int_{A_{r}}  \rho_0,
\]
we may apply Lemma~\ref{lem.rhf} and obtain
$\E_r^\mathrm{rHF}(\eta_r\gamma_0\eta_r)
\le \E_r^\mathrm{rHF}(\widetilde \gamma) + \mathcal{R}$.
By simple computation
\[
\tr \left(-\frac{\Delta}{2} \widetilde \gamma\right) = 2^{3/2}(5\pi^2)^{-1} \int [V_r']_+^{5/2} + 2^{1/2}(3s^2)^{-1} \int [V_r']_+^{3/2},
\] 
we have
\begin{equation*}
\begin{split}
\E_r^\mathrm{rHF}(\widetilde \gamma)
&\le
 \frac{3}{10}(3 \pi^2)^{2/3}\int(\rho_r^\mathrm{TF})^{5/3}  - \int_{A_r} \Phi_{r} \rho_{r}^\mathrm{TF} + D(\rho_r^\mathrm{TF}) \\
&\quad+ Cs^{-2} \int \rho_{r}^\mathrm{TF}
+ \int_{A_{r+s} } (\Phi_{r} -\Phi_{r} \star g^2)\rho_{r}^\mathrm{TF}
+\int_{A_r \cap A_{r+s}^c} \Phi_{r} \rho_r^\mathrm{TF} \\
&= \E_r^\mathrm{TF}(\rho_r^\mathrm{TF}) +Cs^{-2} \int \rho_{r}^\mathrm{TF} +\int_{A_r \cap A_{r+s}^c} \Phi_{r} \rho_{r}^\mathrm{TF},
\end{split}
\end{equation*}
where we have used $ \Phi_r - \Phi_{r} \star g^2 = 0$ on $A_{r+s}$.
This fact follows from the mean value property.
Using Lemma~\ref{lem.outer} and Lemma~\ref{step3}, we have
\[
\int_{A_r \cap A_{r+s}^c} \Phi_{r} \rho_{r}^\mathrm{TF} 
\le Csr^{-8}.
\]
We choose $s=r^{5/3}$ and get
\[
\E_r^\mathrm{rHF}(\widetilde \gamma) \le
\E_{r}^\mathrm{TF}(\rho_{r}^\mathrm{TF}) + Cr^{-7+2/3}.
\]
Finally, since $\lambda \le 1/2$, we have
\begin{align*}
\mathcal{R} \le C(\lambda^{-2}r^{-5} + \lambda r^{-7}),
\end{align*}
which shows the desired upper bound.

\underline{\textit{Lower bound}}
We will prove
\[
\E_r^\mathrm{rHF}(\eta_r \gamma_0\eta_r) \ge \E_{r}^\mathrm{TF}(\rho_{r}^\mathrm{TF}) + D(\eta_r^2 \rho_0 - \rho_{r}^\mathrm{TF}) -Cr^{-7 +1/3}.
\]
As in the proof of Lemma~\ref{thm.initial}, we see
\begin{align*}
\E_r^\mathrm{rHF}(\eta_r \gamma_0\eta_r) 
&= \tr \left[ \left(-\frac{\Delta}{2} - \phi_{r}^\mathrm{TF}\right) \eta_r \gamma_0\eta_r \right]
+D(\eta_r^2\rho_0 - \rho_{r}^\mathrm{TF}) - D( \rho_{r}^\mathrm{TF})\\
&\ge  \E_{r}^\mathrm{TF}(\rho_{r}^\mathrm{TF}) + D(\eta_r^2 \rho_0 - \rho_{r}^\mathrm{TF}) -Cs^{-2}\int \eta_r^2 \rho_0 \\
&\quad -C\left(\int[\phi_{r}^\mathrm{TF}]_+^{5/2} \right)^{3/5} \left( \int[\phi_{r}^\mathrm{TF}- \phi_{r}^\mathrm{TF}\star g^2]_+^{5/2}\right)^{2/5}.
\end{align*}
We note that $|x|^{-1} - |x|^{-1}\star g^2 \ge 0 $ and thus $\rho_{r}^\mathrm{TF} \star(|x|^{-1} - |x|^{-1}\star g^2) \ge0$.
Since the TF equation, we have 
\[
\phi_{r}^\mathrm{TF} - \phi_{r}^\mathrm{TF} \star g^2 \le \1_{A_r}\Phi_{r} - \1_{A_r}\Phi_{r}\star g^2 \eqqcolon f.
\]
By the mean value property, we infer that $\mathrm{supp} f \subset A_{r-s} \cap A_{r+s}^c$ and thus
\[
[\phi_{r}^\mathrm{TF} - \phi_{r}^\mathrm{TF} \star g^2]_+ \le Cr^{-4}\1_{A_{r-s} \cap A_{r+s}^c}.
\]
Together with these facts, we conclude that
\begin{align*}
\E_r^\mathrm{rHF}(\eta_r \gamma_0\eta_r)  &\ge 
 \E_{r}^\mathrm{TF}(\rho_{r}^\mathrm{TF}) + D(\eta_r^2 \rho_0 - \rho_{r}^\mathrm{TF})  -C(s^{-2}r^{-3} + r^{-37/5}s^{2/5}).
\end{align*}
Then  we choose $s = r^{11/6}$ and arrive at the desired lower bound.
After choosing $\lambda = r^{1/3}/2$, the estimate (\ref{eq.step4}) follows.

\underline{\it Conclusion}
Combining the upper and lower bound, we learn
\[
D(\eta_r^2 \rho_0 - \rho_{r}^\mathrm{TF}) \le Cr^{-7}(r^{1/3} + \lambda^{-2}r^2 +\lambda).
\]
Using the Hardy-Littlewood-Sobolev inequality, we have
\begin{align*}
D(\chi_r^+\rho_0 - \eta_r^2 \rho_0)
&\le C\|\1_{A_r \cap A_{(1+\lambda)r}^c}\rho_0\|^2_{L^{6/5}} \\
&\le C\left(\int_{A_r}\rho_0^{5/3}  \right)^{6/5}\left(\sum_{j=1}^K  \int_{r \le |x-R_j| \le (1+\lambda)r} \, dx\right)^{7/15} \\
&=C\lambda^{7/15} r^{-7}.
\end{align*}
By convexity of the Coulomb term $D(\cdot)$, we see
\begin{align*}
D(\chi_r^+\rho_0- \rho_{r}^\mathrm{TF}) &\le 2D(\chi_r^+\rho_0 - \eta_r^2 \rho_0)
+2D(\eta_r^2\rho_0 - \rho_{r}^\mathrm{TF}) \\
&\le Cr^{-7}(\lambda^{7/15} + r^{1/3} + \lambda^{-2}r^2),
\end{align*}
for any $\lambda \in (0, 1/2]$.
Choosing $\lambda = r^{30/37}/2$, we have the upper bound.
\end{proof}

\begin{proof}[Proof of Lemma~\ref{thm.iterative}]
Let $\delta>0$ be a constant sufficiently small and $s \in [r^{1/(1+\delta)}, \min\{r^{\frac{1-\delta}{1+\delta}}, \tilde r\}]$ with $\tilde r = r^{\frac{\xi}{\xi+\eta}}(R_\mathrm{min}/4)^{\frac{\eta}{\xi + \eta}}$.
We split
\begin{align*}
\Phi_{s}(x) - \Phi^\mathrm{TF}_{s}(x) &= \phi_{r}^\mathrm{TF}(x) - \phi^\mathrm{TF}(x) +\int_{A_s}\frac{\rho_{r}^\mathrm{TF}(y) - \rho^\mathrm{TF}(y)}{|x-y|} \, dy\\
&\quad + \sum_{j=1}^K\int_{|y-R_j|<s}\frac{\rho_{r}^\mathrm{TF}(y) - \1_{A_r}(y)\rho_0(y)}{|x-y|} \, dy.
\end{align*}
Using Lemma~\ref{step3}, we have
\begin{align*}
\sup_{\partial A_s} | \phi_{r}^\mathrm{TF}(x) - \phi^\mathrm{TF}(x)|
+ \sup_{\partial A_s}\left| \1_{A_s}(\rho_r^\mathrm{TF} - \rho^\mathrm{TF})\star |x|^{-1} \right| &\le C\left(\frac{r}{s} \right)^\xi s^{-4}.
\end{align*}
The Coulomb estimate~\cite[Lem.~6.4]{TFDWIC} and Lemma~\ref{step4} lead to that for any $x \in \partial B(R_j, s)^c$
\begin{align*}
\left|\1_{B(R_j, s)}(\rho_r^\mathrm{TF} - \1_{A_r} \rho_0)\star |x|^{-1} \right|
&\le C\|\rho_r^\mathrm{TF} - \1_{A_r} \rho_0 \|^{5/6}_{L^{5/3}}\left(sD\left[\1_{A_r} \rho_0 - \rho_r^\mathrm{TF}\right]  \right)^{1/12} \\
&\le Cs^{-4}\left(\frac{s}{r} \right)^4 r^{\epsilon/12}.
\end{align*}
Since $s^{2\delta/(1-\delta)} \le r/s \le s^{\delta}$, we have the lemma.
\end{proof}

\begin{proof}[Proof of Lemma~\ref{lem.potential}]
The following proof is the same as in~\cite[Thm.~7.1]{GotorHF} and~\cite[Thm.~5.1]{Samojlow}.
By Lemma~\ref{thm.initial}, there are constants $C_3>0$ and $\epsilon>0$ such that $(r, C_3, \epsilon) \in \mathcal{A}$ for any $r \le z_\mathrm{min}^{-1/3}$.
Let $\delta>0$ be a constant small enough, $\sigma = \max\{C_2, C_3\}$ and $D_0 = z_\mathrm{min}^{-1/3}$, where $C_2$ is defined in Lemma~\ref{thm.iterative}. 
Now we define for $\epsilon_0>0$ sufficiently small 
\[
M \coloneqq \sup\left\{ r \in \R \colon \sup_{x \in \partial A_s} \left|\Phi_{s}(x) - \Phi_{s}^\mathrm{TF}(x) \right| \le \sigma s^{-4+\epsilon_0}, \text{ for any } s \le r^{\frac{1}{1+\delta}} \right\}.
\]
Next, we suppose that (1) $M < R_\mathrm{min}/4$, and (2) $(M^{\frac{1}{1+\delta}}, \min\{ M^{\frac{1-\delta}{1+\delta}}, \tilde M\}) \neq \emptyset$,
where $\tilde M \coloneqq M^{\xi/(\xi + \eta)} (R_\mathrm{min}/4)^{\eta/(\xi + \eta)}$.
If $D_0 <M$, then there is a sequence  such that $D_{n} \to M$ and $D_{0}  \le D_n \le M$ for large $n$.
From this and Lemma~\ref{thm.iterative}, we see
\[
\sup_{x \in \partial A_r} \left|\Phi_{r}(x) - \Phi_{r}^\mathrm{TF}(x)\right| \le \sigma r^{-4+\epsilon_0}, \quad \text{for any } r \in \left[D_{n}^{\frac{1}{1+\delta}}, \min \left\{D_{n}^{\frac{1-\delta}{1+\delta}}, \tilde D_{n}\right\}\right],
\]
where $\tilde D_n \coloneqq D_n^{\xi/(\xi + \eta)} (R_\mathrm{min}/4)^{\eta/(\xi + \eta)}$.
From (2), we have
\[
M^{\frac{1}{1+\delta}} \in \left(D_{n}^{\frac{1}{1+\delta}}, \min \left\{D_{n}^{\frac{1-\delta}{1+\delta}}, \tilde D_{n}\right\}\right) \neq \emptyset
\]
for large $n$. This contradicts the definition of $M$.
If $D_0=M$, then $D_0 \le R_\mathrm{min}/4$ and $(r, \sigma, \epsilon_0) \in \mathcal{A}$ for any $r \le \min\{M^{\frac{1-\delta}{1+\delta}}, \tilde M\}$, which also contradicts the definition of $M$.
Finally, if $D_0 > M$ then we can choose $M' \in (M,  D_0)$.
This contradicts $(r, \sigma, \epsilon_0) \in \mathcal{A}$ for any $r \le D_0$.
Hence at least one of (1) and (2) cannot hold.
If (1) is true, then $M \ge R_\mathrm{min}^{\frac{\eta(1+\delta)}{\eta-\delta\xi}}$.
Hence the lemma follows.
\end{proof}

\section{Proof of Theorem~\ref{thm.main}}
\label{proof.main}
The following lemma allows us to control the outside models.

\begin{lemma}
\label{lem.outsides}
We assume that $z_\mathrm{min} \ge \delta_0 z_\mathrm{max}$ for some $\delta_0$, and for $\epsilon_3, \delta_3 >0$ sufficiently small $4 \ge R_\mathrm{min} \ge \delta_3^{-1}z_\mathrm{min}^{-1/3 + \alpha}$ with some $\alpha < 2/231$, and $r = \delta_3 R_\mathrm{min}^{1+\epsilon_3}$.
Then for any $s \le r$ and $j=1, \dots, K$ we have
\begin{enumerate}

\item $\sup_{B(R_j, s)^c}\left|(\rho_{z_j}^\mathrm{TF} - \rho^\mathrm{TF})\1_{B(R_j, s)} \star |x|^{-1} \right| \le Cs^{-4 + \epsilon_4}$,

\item $\sup_{B(R_j, s)^c}\left|(\rho_0 - \rho^\mathrm{TF})\1_{B(R_j, s)} \star |x|^{-1} \right| \le Cs^{-4 + \epsilon_4}$,

\item $\left|\int_{B(R_j, s)}(\rho_{z_j}^\mathrm{TF} - \rho^\mathrm{TF})\right| \le Cs^{-4 + \epsilon_4}$,

\item $\left|\int_{B(R_j, s)}(\rho_0 - \rho^\mathrm{TF})\right| \le Cs^{-4 + \epsilon_4}$,
\end{enumerate}
where $\epsilon_4 >0$ is some constant.
\end{lemma}

\begin{proof}
Let $D_j \coloneqq \min_{i \neq j}|R_i -R_j|/2$ and $\epsilon > 0$ be a small constant.
First, we note that $(s/D_j)^\eta \le Cs^{\epsilon}$, $s^{1+\epsilon} \le R_\mathrm{min}/4$ by $s\le r$, and $r \ge z_\mathrm{min}^{-1/3}$.
Using the Sommerfeld estimate~\cite[Thm.~4.1~\&~4.2]{GotorHF}, we see that for any $x \in \partial B(R_j, s)$
\begin{align*}
\phi^\mathrm{TF}(x) - \phi_{z_j}^\mathrm{TF}(x)  &\le c_\mathrm{S} s^{-4}\left( c_1\left( \frac{s}{D_j} \right)^\eta + c_2\left( \frac{s^{1+\epsilon} }{|x-R_j|}\right)^\xi \right) \\
&\eqqcolon \phi_M(x),
\end{align*}
where $c_1, c_2 >0$ are some constants.
We recall $\phi^\mathrm{TF} \le \sum_{j=1}^K \phi^\mathrm{TF}_{z_j}$~\cite[Cor.~3.6]{LiebTF}.
Hence for $\delta >0$ sufficiently small we have $\phi^\mathrm{TF} - \phi_{z_j}^\mathrm{TF} \le \phi_M$ in $\overline{B(R_j, \delta)}$.
Then the maximum principle implies that $\phi^\mathrm{TF} - \phi_{z_j}^\mathrm{TF} \le \phi_M$ in $\overline{B(R_j, s)}$.
Since $(1+t)^{3/2}-1 \le 3t/2+3t^{3/2}/2$ for $t \ge 0$, we obtain
\begin{align*}
\rho^\mathrm{TF} - \rho^\mathrm{TF}_{z_j} &= c(\phi_{z_j}^\mathrm{TF})^{3/2}\left( \left(1+ \left(\phi^\mathrm{TF} - \phi^\mathrm{TF}_{z_j} \right)/\phi^\mathrm{TF}_{z_j}\right)^{3/2} -1\right) \\
&\le C\left(\phi_{z_j}^\mathrm{TF}\right)^{1/2}\left(\phi^\mathrm{TF} - \phi^\mathrm{TF}_{z_j} \right) +C \left(\phi^\mathrm{TF} - \phi^\mathrm{TF}_{z_j} \right)^{3/2}.
\end{align*}
Using Newton's theorem, we have for $|x-R_j| = s$
\begin{align*}
\int_{|y-R_j|<s} \frac{\rho^\mathrm{TF}(y) - \rho^\mathrm{TF}_{z_j}(y)}{|x-y|} \, dy 
&\le Cs^{-4 + \epsilon},
\end{align*}
which proves (1).

Next, we split
\begin{align*}
u_j(x) &\coloneqq (\rho_0 - \rho^\mathrm{TF})\1_{B(R_j, s)} \star |x|^{-1} \\
&= \underbrace{(\rho_0 - \rho^\mathrm{TF})\1_{A_s^c} \star |x|^{-1}}_{\eqqcolon u_s(x) } 
-\underbrace{\sum_{i \ne j} (\rho_0 - \rho^\mathrm{TF})\1_{B(R_i, s)} \star |x|^{-1}}_{\eqqcolon u_0(x)}.
\end{align*}
We note that $u_j$ is harmonic on $B(R_j, s)^c$ and thus $|x-R_j| |u_j(x)| \le s\sup_{\partial B(R_j, s)} |u_j|$ for any $x \in B(R_j, s)^c$ by the maximum principle.
Hence we see that for all $j$
\[
\sup_{\partial A_s}|u_0| \le \sup_{\partial A_s} |u_s|  +\frac{s}{R_\mathrm{min} -s}  \sup_{B(R_j, s)^c} |u_j|.
\]
Then we obtain by Lemma~\ref{lem.potential}
\[
\sup_{B(R_j, s)^c}|u_j| \le C\sup_{\partial A_s}|u_s| \le Cs^{-4 +\epsilon},
\]
which shows (2).
Moreover, (3) and (4) are easy consequences of the estimates such as
\[
\lim_{|x| \to \infty}|x|\left| \int_{|y-R_j|<s} \frac{\rho^\mathrm{TF}(y) - \rho^\mathrm{TF}_{z_j}(y)}{|x-y|} \, dy\right| \le Cs^{-4 + \epsilon}.
\]
This completes the proof.
\end{proof}

Let $N_j \coloneqq z_j - \int_{B(R_j, r)} \rho_{z_j}$ and $V_j \coloneqq \1_{B(R_j, r)^c}\Phi_{j, r}$. We  note that
\[
-\Delta \Phi^\mathrm{TF}_{j, r} =4\pi( z_j\delta_j - \rho_{z_j}^\mathrm{TF} \1_{B(R_j, r)}),
\]
where $\delta_j$ is the Dirac measure at $R_j$, and thus
\begin{align*}
\frac{1}{4 \pi} \int_{\R^3} \Phi^\mathrm{TF}_{j, r} (-\Delta \Phi^\mathrm{TF}_{i, r})
&= \frac{z_i z_j}{|R_i - R_j|} - \int_{|x-R_i| < r} \frac{z_j\rho_{z_i}^\mathrm{TF}(x)}{|x-R_j|} \, dx 
-  \int_{|x-R_j| < r} \frac{z_i\rho_{z_j}^\mathrm{TF}(x)}{|x-R_i|} \, dx \\
&\quad + \iint \frac{(\1_{B(R_j, r)}\rho_{z_j}^\mathrm{TF})(x) (\1_{B(R_i, r)}\rho_{z_i}^\mathrm{TF})(y)}{|x-y|} \, dxdy \\
&= 2D(z_i\delta_i - \rho_{z_i}^\mathrm{TF} \1_{B(R_i, r)}, z_j\delta_j - \rho_{z_j}^\mathrm{TF} \1_{B(R_j, r)}) \\
&\eqqcolon \mathcal{Q}_{ij}^\mathrm{TF} .
\end{align*}

Then we can see that $D^\mathrm{TF}$ is determined by the outside TF models as follows.

\begin{lemma}
\label{lem.TF-Born}
Under the same assumptions as in Lemma~\ref{lem.outsides}, there is a constant $\epsilon_5$ such that
\[
\left|D^\mathrm{TF}(\uZ, \uR) -  \left( \E_r^\mathrm{TF}(\rho_r^\mathrm{TF}) - \sum_{j=1}^K E_{V_j}^\mathrm{TF}(N_j) \right)  \right| \le Cr^{-7 + \epsilon_5}.
\]
\end{lemma}

\begin{proof}
\underline{\it Lower bound}.
 Let $\rho_r^{(j)}$ be a minimizer for the TF problem $E^\mathrm{TF}_{V_j}(N_j)$.
We note that for any $\rho$
\begin{equation}
\begin{split}
\label{eq.TF1}
\E^\mathrm{TF}_{V_{\uR}}(\rho) &= \sum_{j=1}^K \E^\mathrm{TF}_{z_j|x-R_j|^{-1}} (\1_{B(R_j, r)} \rho) +  \E^\mathrm{TF}_{r}(\1_{A_r}\rho) \\
&\quad +\int_{A_r} \rho(x)( \rho-\rho_0 )\1_{A_r^c}\star |x|^{-1} \, dx + \sum_{i<j}2D(\rho\1_{B(R_i, r)}, \rho \1_{B(R_j, r)^c}) \\
&\quad - \sum_{i\neq j}  \int_{|x-R_j|<r} z_i|x-R_i|^{-1} \rho(x) \, dx.
\end{split}
\end{equation}
and
\begin{equation}
\begin{split}
\label{eq.TF2}
\E_{z_j/|x-R_j|}^\mathrm{TF}(\rho) &= \E_{z_j/|x-R_j|}^\mathrm{TF}(\rho\1_{B(R_j, r)}) + \E^\mathrm{TF}_{V_j}(\rho\1_{B(R_j, r)^c}) \\
&\quad +2D(\rho\1_{B(R_j, r)^c}, (\rho - \rho_{z_j})\1_{B(R_j, r)}).
\end{split}
\end{equation}

We use (\ref{eq.TF1}) with $\rho = \rho^\mathrm{TF}$ and insert $\rho = \rho^\mathrm{TF} \1_{B(R_j, r)} + \rho_r^{(j)}$ into (\ref{eq.TF2}).
Then since Lemma~\ref{lem.potential} and Lemma~\ref{lem.outsides} we see
\begin{align*}
D^\mathrm{TF}(\uZ, \uR)&\ge \E_{r}^\mathrm{TF}(\rho_r^\mathrm{TF}) -\sum_{j=1}^K E_{V_j}^\mathrm{TF}(N_j) +  \sum_{i<j}\mathcal{Q}_{ij}^\mathrm{TF} -Cr^{-7 + \epsilon_5}.
\end{align*}

\underline{\it Upper bound}.
Inserting $\rho = \sum_{j=1}^K \rho_{z_j}^\mathrm{TF} \1_{B(R_j, r)} + \rho_r^\mathrm{TF}$ into (\ref{eq.TF1}) and $\rho = \rho_{z_j}^\mathrm{TF}$ in (\ref{eq.TF2}), we have
\begin{align*}
D^\mathrm{TF}(\uZ, \uR) \le \E_{r}^\mathrm{TF}(\rho_r^\mathrm{TF}) -\sum_{j=1}^K E_{V_j}^\mathrm{TF}(N_j) + \sum_{i<j} \mathcal{Q}_{ij}^\mathrm{TF} +Cr^{-7 + \epsilon_5},
\end{align*}
where we have used Lemma~\ref{lem.potential} and Lemma~\ref{lem.outsides}.

Finally, we will show that
\begin{equation}
\label{eq.error}
\left|\mathcal{Q}_{ij}^\mathrm{TF}\right| \le Cr^{-7 + \epsilon_5}.
\end{equation}

Let $\Omega_j$ be a set satisfying $B(R_j, R_\mathrm{min}/4) \subset \Omega_j$ and $\Omega_j \subset B(R_i, R_\mathrm{min}/4)^c$ for $i \neq j$.
Now we pick a smooth function $\chi \in  C_c^\infty(\overline{B(R_i, R_\mathrm{min}/4)^c})$ with $\chi =1$ in $\Omega_j$.
Then, since integration by parts (or equivalently, Green's theorem), we have
\begin{align*}
\mathcal{Q}_{ij}^\mathrm{TF}=\int_{\Omega_j} \Phi^\mathrm{TF}_{j, r} (-\Delta \chi\Phi^\mathrm{TF}_{i})
= \int_{\partial \Omega_j} \left(\Phi_{j, r}^\mathrm{TF} \hat n_j \cdot \nabla \Phi_{i, r}^\mathrm{TF} - \Phi_{i, r}^\mathrm{TF} \hat n_j \cdot \nabla\Phi_{j, r}^\mathrm{TF} \right), \\
\end{align*}
where $\hat n_j$ is the outward normal to $\partial \Omega_j$.
We introduce the Poisson kernel $p_r(x, \xi)$ by
\[
p_r(x, \xi) \coloneqq \frac{1}{4\pi r} \frac{|x|^2-r^2}{|x-\xi|^3}.
\]
By harmonicity (see, e.g.,~\cite[Prob.~3.11]{Simon3}), it holds that for $|x-R_i|>r$
\[
\Phi_{i, r}^\mathrm{TF} (x) = \int_{\partial B(R_i, r)} p_r(x-R_i, \xi -R_i) \Phi_{i, r}^\mathrm{TF} (\xi)\,  d\omega(\xi).
\]
By direct computation, we see that
\[
\nabla_x p_r(x, \xi) = p_r(x, \xi) \left(\frac{3(x-\xi)}{|x-\xi|^2} - \frac{2x}{|x|^2 - r^2} \right),
\]
and therefore, in $B(R_j, r)^c$,
\begin{align*}
\left| \nabla\Phi_{i, r}^\mathrm{TF}(x) \right| &\le \frac{2|x-R_i| \left|\Phi_{i, r}^\mathrm{TF}(x)\right|}{|x-R_i|^2 -r^2} + \sup_{\partial B(R_i, r)}\left|\Phi_{i, r}^\mathrm{TF}\right| \int_{\partial B(R_i, r)} \frac{3p_r (x-R_i, \xi - R_i)}{|x-\xi|} \, d\omega(\xi) \\
&\le \frac{Cr}{R_\mathrm{min}^2} \sup_{\partial B(R_i, r)} \left|\Phi_{i, r}^\mathrm{TF}\right| ,
\end{align*}
where we have used $|x-R_i| |\Phi_{i, r}^\mathrm{TF}(x)| \le r\sup_{\partial B(R_i, r)}\left|\Phi_{i, r}^\mathrm{TF}\right|$ for any $|x-R_i| \ge  r$ (see~\cite[Lem.~6.5]{TFDWIC}) and a simple estimate, followed by  $|x-\xi| \ge |x-R_i| - r$ on $|x-R_i| = \xi$,
\[
\int_{\partial B(R_i, r)} \frac{p_r (x-R_i, \xi - R_i)}{|x-\xi|} \, d\omega(\xi) \le  \frac{Cr}{R_\mathrm{min}^2}.
\]
Consequently, we obtain
\begin{align*}
|\mathcal Q^\mathrm{TF}_{ij}|&\le Cr \sup_{\partial B(R_i, r)} \left|\Phi_{i, r}^\mathrm{TF}\right|  \sup_{\partial B(R_j, r)} \left|\Phi_{j, r}^\mathrm{TF}\right| \\
&\le Cr^{-7+\epsilon_6},
\end{align*}
which shows (\ref{eq.error}).
This finishes the proof.
\end{proof}

As in the TF case, we define
\[
\mathcal{Q}_{ij} \coloneqq 2D(z_i\delta_i - \rho_0 (\eta_-^{(i)})^2, z_j\delta_j - \rho_0  (\eta_-^{(j)})^2).
\]

\begin{lemma}
\label{lem.LDA-Born}
Under the same assumptions as in Lemma~\ref{lem.outsides}, there exists $\epsilon_6>0$ such that
\[
\left|D(\uZ, \uR) -  \left( \E_r^\mathrm{TF}(\rho_r^\mathrm{TF}) - \sum_{j=1}^K E_{V_j}^\mathrm{TF}(N_j) \right)  \right| \le Cr^{-7 + \epsilon_6}.
\]
\end{lemma}

\begin{proof}
\underline{\it Lower bound}.
We recall Lemma~\ref{lem.rhf}.
By construction, $\eta_-^{(j)} \coloneqq \1_{B(R_j, r)} \eta_-$ is smooth for all $j=1, \dots, K$, and thus we have
\begin{align*}
\E(\eta_- \gamma_0 \eta_-) &= \sum_{j=1}^K \E(\eta_-^{(j)} \gamma_0 \eta_-^{(j)})+ \sum_{i < j} 2D((\eta_-^{(j)})^2 \rho_0, (\eta_-^{(i)})^2\rho_0 ).
\end{align*}
We note from Lemma~\ref{lem.potential} that inequalities (\ref{ite1})--(\ref{ite5}) hold true.
Applying Lemma~\ref{lem.rhf} and Lemma~\ref{step4}, we see 
\begin{align*}
 \E(\gamma_0)  &\ge \E(\eta_- \gamma_0 \eta_-) + \E_r^\mathrm{rHF}(\eta_r \gamma_0 \eta_r) - \mathcal{R} \\
 &\quad \ge \E_r^\mathrm{TF} (\rho_r^\mathrm{TF}) +\sum_{j=1}^K \E_{z_j/|x-R_j|}(\eta_-^{(j)} \gamma_0 \eta_-^{(j)}) + \sum_{i < j} 2D((\eta_-^{(j)})^2 \rho_0, (\eta_-^{(i)})^2\rho_0 ) \\
 &\quad-\sum_{i \neq j}\int z_j|x-R_j|^{-1}(\eta_-^{(i)} )^2 \rho_0  - C\lambda^{-2}r^{-5}  -Cr^{-7+1/3}.
\end{align*}
We note that $\tr (\eta_-^{(j)} \gamma_0 \eta_-^{(j)}) < z_j$ for all $j =1 ,\dots, K$.
To see this, we use the atomic Sommerfeld bound~\cite[Thm.~5.4]{SolovejIC}, namely,  there is a constant $C >0$ such that
\[
\int_{|x-R_j| > r} \rho^\mathrm{TF}_{z_j}(x) \, dx \ge C^{-1}r^{-3}.
\]
Combining this with Lemma~\ref{lem.outsides}, we see that 
\begin{align*}
z_j - \tr (\eta_-^{(j)} \gamma_0 \eta_-^{(j)}) &\ge \int_{|x-R_j| > r} \rho^\mathrm{TF}_{z_j}(x) \, dx+ \int_{|x-R_j| < r}\left( \rho^\mathrm{TF}_{z_j}(x) - \rho_0(x) \right) \, dx \\
& \ge C^{-1}r^{-3} -Cr^{-3+ \epsilon_4}\\
&>0.
\end{align*}
Then as in the case of the molecules, we have 
\begin{align*}
E_{z_j/|x-R_j|}(z_j) &\le \E_{z_j/|x-R_j|}\left(\eta_-^{(j)} \gamma_0 \eta_-^{(j)}+ \eta_r^{(j)} \gamma_{z_j} \eta_r^{(j)} \right)\\
&\le  \E_{z_j/|x-R_j|}\left(\eta_-^{(j)} \gamma_0 \eta_-^{(j)} \right) + \E_{V_j}^\mathrm{rHF}\left(\eta_r^{(j)} \gamma_{z_j} \eta_r^{(j)} \right) +Cr^{-7+\epsilon_6},
\end{align*}
where we have used  Lemma~\ref{lem.outsides} in the last inequality.
Using Lemma~\ref{step4}, we see
\[
 \E_{V_j}^\mathrm{rHF}\left(\eta_r^{(j)} \gamma_{z_j} \eta_r^{(j)} \right) \le \E_{V_j}^\mathrm{TF}(\rho_r^{(j)}) + Cr^{-7+1/3}.
\]

Then we obtain
\begin{equation}
\begin{split}
 \E(\gamma_0)  +U_{\uR} 
  &\ge  \E_r^\mathrm{TF} (\rho_r^\mathrm{TF}) +\sum_{j=1}^K \left(E_{z_j/|x-R_j|}(z_j)  - \E_{V_j}^\mathrm{TF}(\rho_r^{(j)}) \right) \\
 &\quad + \sum_{i<j}\mathcal{Q}_{ij}  -Cr^{-7+1/3},
 \end{split}
\end{equation}
which shows the lower bound.

\underline{{\it Upper bound.}}
Since Lemma~\ref{lem.rhf} and  Lemma~\ref{step4}--\ref{lem.outsides}, we see
\begin{align*}
E_{V_{\uR}} (Z) + U_{\uR}&\le \E_{V_{\uR}}\left(\sum_{j=1}^K \eta_-^{(j)} \gamma_{z_j} \eta_-^{(j)} + \eta_r \gamma_0 \eta_r \right) + U_{\uR} \\
&\le \sum_{j=1}^K \E_{z_j/|x-R_j|}\left(\eta_-^{(j)} \gamma_{z_j} \eta_-^{(j)} \right) + \E_{r} ( \eta_r \gamma_0 \eta_r ) + \sum_{i<j} \mathcal{Q}_{ij} + Cr^{-7+\epsilon_6} \\
&\le \sum_{j=1}^K \left(E_{z_j/|x-R_j|}(z_j) -E_{V_j}^\mathrm{TF}(N_j) \right) + \E_{r}^\mathrm{TF}(\rho_r^\mathrm{TF}) + \sum_{i<j} \mathcal{Q}_{ij} +  Cr^{-7+\epsilon_6}.
\end{align*}
By copying the proof of (\ref{eq.error}), we can see $\left|\mathcal{Q}_{ij}\right| \le Cr^{-7 + \epsilon_6}$.
Then the proof is complete.
\end{proof}

\begin{proof}[Proof of Theorem~\ref{thm.main}]
First, we assume that $4 \ge R_\mathrm{min} \ge \delta_3^{-1}z_\mathrm{min}^{-1/3+\alpha}$ and $r = \delta_3 R_\mathrm{min}^{1+\epsilon_3}$ as in Lemma~\ref{lem.outsides}.
Combining Lemma~\ref{lem.LDA-Born} and Lemma~\ref{lem.TF-Born}, we have the desired conclusion in this case.
Moreover, we obtain (\ref{thm.BO-curve}).

Next, we consider the case $R_\mathrm{min} \le Z^{-1/3}$.
By (\ref{eq.semicl}) we know $|E_{V_{\uR}} (Z) - E_{V_{\uR}}^\mathrm{TF} (Z)| \le CZ^{7/3 -2/33}$ and $|E_{z_j/|x-R_j|} (z_j) - E_{z_j/|x-R_j|}^\mathrm{TF} (z_j)| \le Cz_j^{7/3 -2/33}$  for all $j =1, \dots, K$.
Let $C_{\uZ} \coloneqq z_\mathrm{max} / z_\mathrm{min}$.
Then it follows that there is a $\epsilon>0$ so that
\begin{align*}
\left| D(\uZ, \uR) - D^\mathrm{TF}(\uZ, \uR) \right| 
&\le C\left( 1+ C_{\uZ}^{7/3 -2/33} \right) R_\mathrm{min}^{-7+\epsilon},
\end{align*}
which shows the conclusion.

Similarly, we deduce from $z_\mathrm{min}^{-1} = C_{\uZ} z_\mathrm{max}^{-1}$ that the desired result  for $Z^{-1/3} \le R_\mathrm{min} \le \delta_3^{-1}z_\mathrm{min}^{-1/3+\alpha}$.
\end{proof}

\begin{proof}[Proof of Corollary~\ref{cor.main}]
Let $E_\mathrm{mol}(\uZ) \coloneqq \inf_{\uR} (E_{V_{\uR}}(Z) + U_{\uR}) $ be the Born-Oppenheimer ground state energy in Kohn-Sham theory.
The following lemma is an elementary property of this energy.

\begin{lemma}
\label{lem.binding}
For any configurations $\uZ_1 = (z_{\pi(1)}, \dots, z_{\pi(p)})$ and $\uZ_2 = (z_{\pi(p+1)}, \dots, z_{\pi(K)})$ with $1 \le p \le K-1$ and $\pi$ permutation of $\{1, \dots, K\}$, we have
\[
E_\mathrm{mol}(\uZ) \le E_\mathrm{mol}(\uZ_1) + E_\mathrm{mol}(\uZ_2).
\]
\end{lemma}

\begin{proof}[Proof of Lemma~\ref{lem.binding}]
Let $\epsilon >0$.
We can take $\gamma_n^{(i)}$ and  $\uR_n^{(i)}$ such that $\tr \gamma_n^{(i)} = |\uZ_i|$, each $\mathrm{supp} \, \rho_{\gamma_n^{(i)}}$ is in a ball of radius $r>0$, and 
\[\E_{V_{\uR_n}^{(i)}}(\gamma_n^{(i)}) +U_{\uR_n^{(i)}} \le E_\mathrm{mol}(\uZ_i) + 1/n.
\] 
For $r_n \in \R^3$ we define $\gamma_n^{(3)} \coloneqq \tau_{-r_n} \gamma_n^{(2)} \tau_{r_n}$ with $\tau$ being the translation operator, and $\gamma_n \coloneqq \gamma_n^{(1)} + \gamma_n^{(3)}$.
Then we see that $0 \le \gamma_n \le 1$, $\tr \gamma_n =Z$, and $\mathrm{supp} \, \rho_{\gamma_n^{(1)}} \cap \mathrm{supp} \, \rho_{\gamma_n^{(3)}} = \emptyset$ for large $|r_n|>0$.
Let $\uR_n \coloneqq (R_n^{(\pi(1))}, \dots, R_n^{(\pi(p))}, R_n^{(\pi(p+1))} + r_n, R_n^{(\pi(p+2))} +r_n, \dots, R_n^{(\pi(K))} + r_n)$ with $|r_n| > n$.
By simple computation, we have
\[
2D( \rho_{\gamma_n^{(1)}}, \rho_{\gamma_n^{(3)}}) \le \frac{Z^2}{n-2r},
\]
and hence
\begin{align*}
E_\mathrm{mol}(\uZ) &\le \E_{V_{\uR_n}}(\gamma_n) + U_{\uR_n} \\
&\le  E_\mathrm{mol}(\uZ_1) + E_\mathrm{mol}(\uZ_2)  + \epsilon
\end{align*}
for large $n$.
\end{proof}

Now we assume that there exists $\uRo$ such that $E_\mathrm{mol}(\uZ) = E_{V_{\uRo}}(Z) + U_{\uRo}$.
Let $R_\mathrm{M} \coloneqq \min_{i \neq j}|R_0^{(i)} - R_0^{(j)}|$.
With Lemma~\ref{thm.initial} and Lemma~\ref{lem.binding}, it follows that
\[
0 \ge E_{V_{\uRo}}(Z)  + U_{\uRo} \ge -C_3Z^{7/3} +\frac{Z^2}{C_3R_\mathrm{M}},
\]
and thus $R_\mathrm{M} \ge C_3^{-2}Z^{-1/3}$.
Then we have $D^\mathrm{TF}(\uZ, \uRo) \ge C_4R_\mathrm{M}^{-7}$ (see the proof of~\cite[Thm.~8]{RuskaiSolovej}).
Without loss of generality we can assume $z_\mathrm{min} \ge 1$.
Using Theorem~\ref{thm.main} and Lemma~\ref{lem.binding}, we have
\[
0 \ge D(\uZ, \uRo) \ge C_5^{-1}R_\mathrm{M}^{-7}- C_5R_\mathrm{M}^{-7+\epsilon}.
\]
This completes the proof.
\end{proof}
\section*{Acknowledgments}
The author wishes to express her thanks to Tetsuo Hatsuda, Tomoya Naito, Shu Nakamura, and Takeru Yokota for helpful comments concerning the introduction.
She also thanks the anonymous referees for pointing out many errors and for helpful suggestions that improved the paper.

\end{document}